\documentclass[a4paper,11pt]{article}
\pdfoutput=1

\usepackage[utf8]{inputenc}
\usepackage{jheppub}
\usepackage[shortlabels]{enumitem}
\usepackage{amssymb,amsmath,amsbsy,amsthm,mathtools, physics}
\usepackage{braket}

\newtheorem{proposition}{Proposition}

\newtheorem{lemma}[proposition]{Lemma}
\newtheorem{theorem}[proposition]{Theorem}
\newtheorem{corollary}[proposition]{Corollary}
\theoremstyle{remark}

\usepackage{graphicx,color}
\usepackage{multirow}
\usepackage{booktabs}
\usepackage{hyperref}
\usepackage{subcaption}
\newcommand{\lb}{\left[}
\newcommand{\rb}{\right]}
\newcommand{\lp}{\left(}
\newcommand{\rp}{\right)}

\newcommand{\mB}{\mathcal{B}}

\newcommand{\mO}{\mathcal{O}}

\newcommand{\mbI}{\mathbb{I}}

\newcommand{\ep}{\epsilon}

\newcommand{\Renyi}{R\'{e}nyi }
\newcommand{\Renyitwo}{R\'{e}nyi-2 }

\newcommand{\Ball}[2]{\mB_{#1}\lp #2 \rp}

\title{The Capacity of Entanglement and\\ Holographic Entropies at Finite Resources}

\author[a,b]{Ning Bao}
\author[a]{Jacob March}

\affiliation[a]{Department of Physics, Northeastern University, Boston, MA, 02115, USA}
\affiliation[b]{Computational Science Initiative, Brookhaven National Laboratory, Upton, NY 11973 USA}

\emailAdd{ningbao75@gmail.com}
\emailAdd{j.march@northeastern.edu}

\abstract{
    The Ryu-Takayanagi formula equates the area of a minimal surface with the von Neumann entropy of a boundary subregion, and leaves two things about that identification open.
    The first is how sharply a geometry fixes an entropy. Fannes-Audenaert answers with a Hilbert-space dimension, which diverges as the cutoff is removed however close the two states are. We replace it with the capacity of entanglement, the variance of the modular energy, whose square root grows like the square root of the entangling area where the dimensional factor grows like the regulated volume. The bound is dimension-free and saturated, and it makes the ambiguity of the entropy subextensive for any perturbation whose capacity is small compared with $S_{vN}^2$ times the trace norm.
    The second is what the area means for a single state, since compression and dilution rates are defined only for many copies while a geometry describes one. When a single replica saddle dominates near $\alpha = 1$, every smooth \Renyi entropy at fixed $\alpha > 1$ agrees with $S_{vN}$ to $\mO(\sqrt{S_{vN}})$, as do the smooth min- and max-entropies. The minimal surface therefore fixes every one-shot entropy of the state at once, with large central charge playing the role of large copy number in the asymptotic equipartition property.
    As a consequence we bound how far outside the holographic entropy cone a holographic state can appear to fall, leaving estimation and certification open.
}

\makeatletter
\gdef\@fpheader{}
\makeatother

\begin{document}

\maketitle

\section{Introduction}
\label{sec:intro}

One of the most profound insights offered by the correspondence between anti-de Sitter space and conformal field theory (AdS/CFT) \cite{Maldacena:1997re} is the realization that the entanglement structure of a holographic state is encoded in the semiclassical geometry of its gravitational dual.
For static spacetimes, this connection is quantified by the Ryu-Takayanagi (RT) formula \cite{Ryu:2006bv, Ryu:2006ef}.
To leading order in the central charge, the von Neumann entropy
\begin{equation}
    S_{vN}(\rho) = -\Tr(\rho \log \rho)
    \label{vN_entropy}
\end{equation}
of a boundary subregion $A$ is given by
\begin{equation}
    S_{vN}(\rho_A) = \min_{\gamma \sim A} \frac{\text{Area}(\gamma)}{4 G_N},
    \label{rt}
\end{equation}
where $\gamma$ is a codimension-2 minimal surface in the bulk homologous to $A$.
This prescription has been robustly motivated by gravitational path integral arguments \cite{Lewkowycz:2013nqa, Dong:2016hjy}, extended to time-dependent settings \cite{Hubeny:2007xt, Wall:2012uf}, and, more recently, obtained for classes of multi-boundary states via CFT ensembles in the context of $\text{AdS}_3/\text{CFT}_2$ \cite{Bao:2025plr}.

The quantity on the left of \eqref{rt} is, however, an asymptotic one.
Its operational meaning, as an optimal compression rate or an optimal entanglement dilution rate, is tied to asymptotically many identical copies of $\rho_A$, a resource that is never available in a physical or numerical experiment.
What an experiment can measure are the integer \Renyi moments $\Tr \rho^\alpha$, accessible through SWAP-type protocols and randomized measurements \cite{Ekert:2002qtj, Hastings:2010zka, Cardy:2011zz, vanEnk:2011xlo, Humeniuk:2012xg, Abanin:2012jms, Daley:2012xhf, Islam:2015mom, Linke:2017xlv, Brydges:2019wut, Elben:2022jvo}, which determine the \Renyi entropies
\begin{equation}
    S_\alpha(\rho) := \frac{1}{1-\alpha} \log \Tr (\rho^\alpha)
    \label{renyi_def}
\end{equation}
at integer index $\alpha \ge 2$.
These are not substitutes for \eqref{vN_entropy}.
For a single interval of length $\ell$ in the vacuum of a holographic CFT$_2$ with UV cutoff $a$ \cite{Holzhey:1994we, Calabrese:2004eu},
\begin{equation}
    S_\alpha = \frac{c}{6}\lp 1 + \frac{1}{\alpha} \rp \log\frac{\ell}{a},
    \label{cardy_renyi}
\end{equation}
so the bare \Renyitwo entropy is $S_2 = \tfrac{3}{4} S_{vN}$.
The difference is extensive, of the same order of central charge as the area term itself.

Additionally, a laboratory preparation of $\rho$ is only ever known to lie within a trace-norm ball of small radius $\ep$, written $\Ball{\ep}{\rho}$, set by the preparation fidelity, since even for two known states, distinguishing them at trace distance $\ep$ with fixed confidence requires of order $1/\ep^2$ copies \cite{Helstrom1969Detection, AudenaertQuantumChernoff2007}.
Over that ball the von Neumann entropy can vary by up to $\tfrac{\ep}{2}\log (d-1) + h(\ep/2)$, where $d$ is the dimension of the regulated Hilbert space of the subregion and $h$ the binary entropy \cite{Audenaert:2006vjl}. Throughout, $\ep$ and $\delta$ denote trace norms, twice the trace distance in which \cite{Audenaert:2006vjl} is stated.
This bound is useless in the context of an infinite dimensional CFT.
The von Neumann entropy of a prepared state is therefore not a well-posed quantity at finite resources, and neither is the bare \Renyi entropy at the resolution that matters, since the moments are exponentially small and states in the ball can differ in their bare \Renyi entropies at leading order in $c$.

The first result of this paper settles how sharply a geometry fixes an entropy.
Theorem \ref{thm:continuity} is a continuity bound for the von Neumann entropy itself, in which the dimensional factor of the Fannes-Audenaert inequality \cite{Audenaert:2006vjl} is replaced by the capacities of entanglement of the two states compared, the variance of the modular energy defined in \eqref{capacity_def} below.
It is dimension-free, it vanishes as the trace norm does, and it holds at every trace norm below its maximal value of two. We exhibit pairs of states that saturate it, so no bound of the form $c_1(\delta)\lp\sqrt{C_\rho}+\sqrt{C_\sigma}\rp + c_2(\delta)$ has a smaller $c_1$ or $c_2$.
The proof of theorem \ref{thm:continuity} uses no one-shot machinery at all, only Mirsky's inequality \cite{Mirsky1960SymmetricGauge}, an exact identity for the entropy change at each eigenvalue, and one application each of Jensen's inequality and Cauchy-Schwarz.
It matters in a continuum theory because $\sqrt{C}$ grows like the square root of an area where $\log d$ grows like a volume, quantified in section \ref{sec:continuity}.
It turns the $\mO(\ep \log d)$ ambiguity of section \ref{sec:ball} into $\mO(\sqrt{\ep\, C_\sigma})$ against any particular perturbation $\sigma$, which is subextensive whenever $C_\sigma$ is small compared with $S_{vN}^2/\ep$.

Our second result concerns the quantities an experiment actually reports, where the same modular energy governs the answer.
The replacements that are well posed over the ball are the smooth entropies of one-shot information theory \cite{Renner2004SmoothRE, Renner:2005qbj}, the extremal values a given entropy takes over the ball.
The smooth \Renyi entropy at $\alpha > 1$ is
\begin{equation}
    S^\ep_\alpha(\rho) := \max_{\sigma \in \Ball{\ep}{\rho}} S_\alpha(\sigma),
    \label{smooth_entropy_def}
\end{equation}
by definition the largest value the bare $S_\alpha$ can take on any state $\ep$-indistinguishable from $\rho$.
For holographic states, at each fixed $\alpha > 1$,
\begin{equation}
    S^\ep_\alpha(\rho) = S_{vN}(\rho) + \mO\lp\sqrt{S_{vN}(\rho)}\rp
    \label{main_result}
\end{equation}
at fixed smoothing parameter $\ep \in (0,1)$.
The mechanism is concentration of the modular energy $E = -\log\lambda$, where $\lambda$ is an eigenvalue of $\rho$ drawn with probability $\lambda$.
$E$ has mean $S_{vN}$ and variance the same capacity $C$, which for a holographic state is $\mO(S_{vN})$.
The relative width of the modular spectrum therefore vanishes as $1/\sqrt{S_{vN}}$, and the large central charge limit plays the role the copy number plays in the quantum asymptotic equipartition property \cite{Renner:2005qbj, Tomamichel:2009ohu}.
The statement holds at every fixed index, theorem \ref{thm:ceiling}, whose proof uses only those two moments through a modular window and one application of H\"older's inequality for Schatten norms, and corollary \ref{cor:flatholo} specializes it to a holographic subregion.
It is a single-shot statement about one copy of one state, with explicit finite constants, where the counterpart in \cite{NuradhaWilde2023Fidelity} is asymptotic.


The holographic entropy cone \cite{Bao:2015bfa} is the polyhedral cone of subsystem entropy vectors realizable by the RT formula, a proper subcone of the quantum entropy cone for three or more parties that certifies certain entanglement patterns, GHZ-type correlations among them, to be incompatible with a semiclassical bulk.
The inequalities of the cone constrain the von Neumann entropies at leading order in $1/G_N$ and no better, and no substitute cone is available off the shelf at fixed index, since for $0 < \alpha < 1$ the \Renyi entropies of general quantum states satisfy no nontrivial linear inequalities beyond nonnegativity, and for $\alpha > 1$ none that are homogeneous \cite{Linden:2012kdb}.
Those statements concern all quantum states rather than the holographic subclass, so they do not by themselves exclude a fixed-index cone for holographic states, and we are not aware of one.
A test of this kind must therefore land on $S_{vN}$ itself at $\mO(c)$ resolution, which is the accuracy at which \eqref{main_result} operates.
Whether such a test can be assembled from measured moments is left open here, for reasons recorded in section \ref{sec:discussion}.

Proposition \ref{prop:subgauss} with corollary \ref{cor-window} is new.
That the smooth entropies of holographic states flatten onto $S_{vN}$ was established for the min- and max-entropies, the $\alpha \to \infty$ and $\alpha \to 0$ endpoints of the \Renyi family \eqref{renyi_def}, for single intervals in two dimensions in \cite{Czech:2014tva}, and for arbitrary regions in arbitrary dimension in unpublished work of Hayden, Swingle and Walter \cite{Hayden:unpublished} reported in \cite{Bao:2018pvs}, whose projected-state construction underlies our window lemma.
That construction carries a window of half-width of order $\sqrt{\log(1/\ep)}$ in units of $1/G_N$, obtained there by saddle point. The explicit constant \eqref{gaussian_window} is a sharpening of it and is proved here.
It is proved here, in arbitrary dimension and with the range of smoothing parameters over which it holds made explicit, from the leading-order form \eqref{renyi_leading_order} together with three times differentiability of $s_\alpha$ on a fixed neighborhood of $\alpha = 1$.
For a single interval in two dimensions the same statement follows from the exact modular spectrum \cite{Czech:2014tva}, which no other region supplies.
Theorem \ref{thm:continuity} is new as well, and concerns the von Neumann entropy directly rather than any smooth entropy.
Dimension-free continuity bounds that involve a moment of the state are not new in themselves.
The closest precedent is classical. Theorem 1 of \cite{Cohen2021DimensionFree}, presented there as a dimension-free analog of the Cover-Thomas continuity bound, controls $\abs{H(\mu)-H(\nu)}$ by the total variation distance together with the uncentered information moments $\sum_i \xi_i \abs{\log \xi_i}^\alpha$ of the two distributions. At $\alpha = 2$ that moment is $\expval{E^2} = C + S_{vN}^2$, so on a holographic spectrum its square root is $S_{vN}$ rather than $\sqrt{C}$ and the resulting bound is extensive. Centering the moment keeps \eqref{capacity_continuity} subextensive, which is the difference that matters here. Their bound also carries a dependence on $\norm{\mu-\nu}_\infty$, where \eqref{capacity_continuity} uses the trace norm alone, and it is classical where \eqref{capacity_continuity} is a statement about density matrices with matching saturating families.
Winter's energy-constrained continuity \cite{Winter:2015qhs} is the other precedent, and there the difference lies in how the moment enters. A constraint picks out a class of states, where theorem \ref{thm:continuity} applies to any pair and carries the capacity of each of the two inside the bound.
A capacity constraint would in any case lie outside the Alicki-Fannes-Winter method those bounds use, which mixes the two states and so needs a convex constraint set. Bounded capacity is not convex, since two flat states of vanishing capacity can mix to one of positive capacity.
A separate recent line improves the dimensional factor of Fannes-Audenaert rather than removing it \cite{Audenaert2025ContinuityBounds, Berta2025IntegralRepresentations}, and section \ref{sec:continuity} compares those bounds with theorem \ref{thm:continuity}.
Theorem \ref{thm:ceiling} assembles two ingredients that are known separately.
The bounding step, that a state placing weight $1-\ep$ on a subspace of dimension $r$ has $S_\alpha \le \log r + \tfrac{\alpha}{\alpha-1}\log\tfrac{1}{1-\ep}$, is a specialization of a known inequality of \cite{KhatriWilde2024Principles}, set out in section \ref{sec:flatness}.
The rank bound, that such a subspace exists with $\log r$ at most $S_{vN}$ plus a half-width of order $\sqrt{C}$, is in substance result 2 of \cite{Boes:2020vpv}, with the same window, though their route is through Lorenz curves rather than the one-sided Chebyshev bound used here. The projected-state construction it rests on is that of \cite{Hayden:unpublished, Bao:2018pvs}.
The assembly is what is new.
The inequalities of \cite{Boes:2020vpv} cannot simply be chained, their optimization over the ball running in the opposite direction to ours, as recorded around \eqref{bnw}. What is reused is the projector their proof constructs, not their conclusion, and the work is in showing that one projector, fixed before any perturbation is chosen, suffices for the whole ball.

The collapse at fixed index is known outside holography in the i.i.d. setting.
Corollary 4 of \cite{NuradhaWilde2023Fidelity} expands the smooth sandwiched \Renyi relative entropy of $n$ copies to second order in $1/\sqrt n$, at every $\alpha > 1$ and against an arbitrary reference operator. Taken against the identity, where that relative entropy is $-S_{vN}$ and its variance is the capacity, their expansion reads
\begin{equation}
    S^\ep_\alpha(\rho^{\otimes n}) \;=\; n S_{vN} + \sqrt{nC}\,\Phi^{-1}(\ep) + \mO(\log n) ,
\end{equation}
with $\Phi$ the standard normal distribution function. They observe that at this order every index $\alpha > 1$ gives the same answer, agreeing with the $\alpha \to \infty$ member of the family. That is \eqref{sharp_constant} below in the i.i.d. setting and with the sharp constant.
Their smoothing is over a fidelity ball of subnormalized states rather than the normalized trace-norm ball \eqref{ball_def}, so the constants require the translation footnote \ref{fn:conventions} warns about, but the $\sqrt{C}$ scaling and the independence of the index are the same.
On the entropy side rather than the relative-entropy side, \cite{SakaiTan2020Smooth} derive the corresponding classical expansion, restricting throughout to $0 < \alpha < 1$.
What theorem \ref{thm:ceiling} supplies that these do not is a single-shot statement, valid for one copy of one arbitrary state with explicit finite constants.
That is the only version a holographic subregion can use, since there is no tensor power available and the concentration has to come from large central charge at a single copy.

One-shot quantities have been studied in holography in a different direction by \cite{Akers:2023fqr}, which constructs covariant max- and min-entanglement wedges and a one-shot generalized second law.
Those results and ours share a motivation but no technical content.
Closest in spirit is \cite{Wang:2021ptw}, which starts from the same tension between single-copy reconstruction and an asymptotic entropy but resolves it with replica copies in fixed-area states rather than with the large central charge limit at a single copy in generic ones.


Section \ref{sec:setup} collects the modular structure of a holographic subregion, the two moments that carry the argument, the modular density of states and the sub-Gaussian window it supports, proposition \ref{prop:subgauss} with corollary \ref{cor-window}, and the preparation ball.
Section \ref{sec:continuity} proves the continuity bound, which uses none of the one-shot machinery.
Section \ref{sec:flatness} introduces the smooth entropies, proves the flatness of the smooth \Renyi spectrum at every fixed index, and shows that mixture noise cannot reproduce it once the contaminant is separated from the target.
Section \ref{sec:discussion} considers implications for holographic systems and future directions.

Throughout, the Hilbert space of the subregion is finite dimensional, as it is after any lattice or short-distance regulation.
This is used whenever a rank is counted or a maximum over the ball is attained.

\section{Modular Structure and the Preparation Ball}
\label{sec:setup}

\subsection{Modular energy and its two moments}
\label{sec:modular}

Let $\rho$ be the reduced state of a subregion, let $K = -\log \rho$ be its modular Hamiltonian, with eigenvalues $E_i = -\log \lambda_i$, and let
\begin{equation}
    Z(\alpha) = \Tr \rho^\alpha = e^{(1-\alpha) S_\alpha}
    \label{partition_function}
\end{equation}
denote the \Renyi partition function.
For a holographic state the \Renyi entropies take the leading-order form
\begin{equation}
    S_\alpha = \frac{s_\alpha}{G_N} + \mO(1),
    \label{renyi_leading_order}
\end{equation}
where $s_\alpha$ is independent of $G_N$, depends smoothly and nontrivially on $\alpha$, and is computed by the areas of backreacting cosmic branes \cite{Headrick:2010zt, Dong:2016fnf}.
Several statements below differentiate the remainder rather than merely bound it, and that requires more than \eqref{renyi_leading_order} supplies, since a remainder of order one pointwise in $\alpha$ need not have controlled derivatives.
We assume \textit{local smoothness}, that on a fixed neighborhood $\abs{\alpha - 1} \le t_0$ the remainder in \eqref{renyi_leading_order} is three times differentiable with first three derivatives of order one, uniformly in $G_N$. Together with the smoothness of $s_\alpha$ this makes the cumulant generating function \eqref{cgf} three times differentiable there, which is the form the assumption is used in below.
This is used at \eqref{capacity_def}, where the capacity is read off by differentiating twice, in section \ref{sec:dos}, in proposition \ref{prop:subgauss}, and through those in corollary \ref{cor:flatholo}. It is a statement about the subleading expansion and is assumed rather than derived. Nothing below needs more than three derivatives.
This is the sole holographic input in the paper. It carries the sharpened window of section \ref{sec:dos} and the holographic readings of theorems \ref{thm:continuity} and \ref{thm:ceiling}, entering through the two moments below together with local smoothness. Both theorems themselves hold for arbitrary states.

Treat the modular energy as the random variable obtained by drawing an eigenvector of $\rho$ with probability equal to its eigenvalue $\lambda_i$ and recording $E_i$.
Its mean is the von Neumann entropy,
\begin{equation}
    \expval{E} = -\sum_i \lambda_i \log \lambda_i = S_{vN},
    \label{modular_mean}
\end{equation}
and its variance is the capacity of entanglement,
\begin{equation}
    C := \expval{E^2} - \expval{E}^2 = \partial_\alpha^2 \log Z(\alpha) \big|_{\alpha=1} = \mO(1/G_N),
    \label{capacity_def}
\end{equation}
where the scaling follows from \eqref{renyi_leading_order} with local smoothness, since $\log Z(\alpha) = (1-\alpha)s_\alpha/G_N$ with $s_\alpha$ smooth and the remainder contributes at order one after two derivatives.
Thus $S_{vN}$ is of order $1/G_N$, and $C = \mO(1/G_N)$ with $C \gtrsim 1/G_N$ generically as well.
On special states the capacity can vanish outright or be parametrically smaller, a flat spectrum having $C = 0$ while $S_{vN} = \log d$.
The ceilings of section \ref{sec:flatness} use only the upper bound. Proposition \ref{prop:subgauss} also needs the lower one, so that the $\mO(\sqrt{C})$ half-width dominates a displacement independent of $G_N$.
The capacity is one of a family of measures of how far a spectrum departs from flatness, whose relations to one another are set out in \cite{Jasser2026Antiflatness}.
The relative width $\sqrt{C}/S_{vN} = \mO(1/\sqrt{S_{vN}})$ vanishes at large central charge. That concentration is what the flatness results of section \ref{sec:flatness} rest on. Theorem \ref{thm:continuity} needs none of it, holding for arbitrary states, and uses the capacity only as it appears in the bound.
Every bound below is assembled from the elementary inequalities collected here, together with elementary properties of $f(x) = -x\log x$ recorded where they are used.
For a nonnegative random variable $X$ and $a > 0$, and for a real random variable of mean $\expval{X}$ and variance $V$,
\begin{align}
    \Pr(X \ge a) &\;\le\; \frac{\expval{X}}{a} ,
    \label{ineq_markov} \\
    \Pr\lp X - \expval{X} \ge W \rp &\;\le\; \frac{V}{V + W^2} ,
    \label{ineq_cantelli} \\
    \abs{\expval{XY}} &\;\le\; \sqrt{\expval{X^2}}\;\sqrt{\expval{Y^2}} ,
    \label{ineq_cs} \\
    \expval{\varphi(X)} &\;\le\; \varphi\lp\expval{X}\rp \quad \text{for concave } \varphi ,
    \label{ineq_jensen}
\end{align}
the first due to Markov, the second the one-sided Chebyshev inequality of Cantelli \cite{Cantelli1928Confini}, and the last two
Cauchy-Schwarz and Jensen for expectations, for all of which \cite{Boucheron2013Concentration} is a convenient reference.
For operators $A$ and $B$, states $\rho$ and $\sigma$, and any $P$ with $0 \le P \le 1$, writing $\lambda^\downarrow(\rho)$ for the eigenvalues of $\rho$ in decreasing order,
\begin{align}
    \abs{\Tr(AB)} &\;\le\; \norm{A}_{p} \norm{B}_{q} , \qquad \tfrac1p + \tfrac1q = 1 ,
    \label{ineq_holder} \\
    \abs{\Tr\lb P(\rho - \sigma) \rb} &\;\le\; \tfrac{1}{2}\norm{\rho-\sigma}_1 ,
    \label{ineq_proj} \\
    \norm{\lambda^\downarrow(\rho) - \lambda^\downarrow(\sigma)}_1 &\;\le\; \norm{\rho-\sigma}_1 ,
    \label{ineq_mirsky} \\
    \lambda^\downarrow_i(A) &\;\ge\; \lambda^\downarrow_i(B) \quad \text{whenever } A \ge B \ge 0 ,
    \label{ineq_weyl}
\end{align}
which are H\"{o}lder for Schatten norms at any conjugate pair $p,q \in [1,\infty]$ \cite{Bhatia1997Matrix}, the statement that trace distance controls expectation values of projectors, the Lidskii-Mirsky-Wielandt theorem \cite{Mirsky1960SymmetricGauge, Bhatia1997Matrix}, and Weyl monotonicity.

The last inequality needed differs from these in concerning a mixture of two states rather than a comparison between them. On positive operators the map $A \mapsto \Tr A^\alpha$ is convex, so
\begin{equation}
    \Tr\lp \theta A + (1-\theta)B \rp^\alpha \;\le\; \theta \Tr A^\alpha + (1-\theta) \Tr B^\alpha ,
    \label{ineq_convex}
\end{equation}
for any $\theta \in [0,1]$ and any $\alpha \ge 1$.
Throughout, the weight of a set of eigenvalues means the total probability $\sum_i \lambda_i$ it carries.
Section \ref{sec:dos} makes the resemblance to a sum of independent variables precise.

\subsection{The modular density of states at large central charge}
\label{sec:dos}

Theorems \ref{thm:continuity} and \ref{thm:ceiling} require only the two moments above, but for a holographic state the full distribution of the modular energy is available, and it sharpens the concentration estimates those bounds are stated in terms of.
The holographic input enters through the cumulants of the modular energy, and the tail bound proved below is sharper than one built from two moments alone \cite{Bao:2018pvs}.

Introduce the density of modular states $D(E) = \sum_i \delta(E - E_i)$, in terms of which the \Renyi partition function \eqref{partition_function} is a Laplace transform,
\begin{equation}
    Z(\alpha) = \Tr \rho^\alpha = \int_0^\infty dE \, D(E) \, e^{-\alpha E} = e^{(1-\alpha) S_\alpha} .
    \label{app_laplace}
\end{equation}
Given \eqref{renyi_leading_order}, the inverse transform may be evaluated by saddle point at small $G_N$, yielding a density of states \begin{equation}
    D(E) = e^{f(G_N E)/G_N + o(1/G_N)},
    \label{dos_saddle}
\end{equation}
where the $\mO(1)$ rate function $f$ is determined by $(1-\alpha)s_\alpha$ through a Legendre transform \cite{Bao:2018pvs}\footnote{\label{fn:legendre}Explicitly, writing $e = G_N E$ and $D(E) = e^{f(e)/G_N}$, the saddle of \eqref{app_laplace} reads $\sup_e \lb f(e) - \alpha e \rb = (1-\alpha)s_\alpha$, so that $f$ is recovered as $f(e) = \inf_\alpha \lb (1-\alpha)s_\alpha + \alpha e \rb$. For a vacuum interval, where $s_\alpha = \tfrac{s_1}{2}\lp 1 + 1/\alpha \rp$, the infimum sits at $\alpha = \sqrt{s_1/(2e - s_1)}$ and evaluates to $f(e) = \sqrt{s_1(2e - s_1)}$, that is $D(E) = \exp\sqrt{S_{vN}(2E - S_{vN})}$, which is the entanglement spectrum of \cite{CalabreseLefevre2008Spectrum}. Its support begins at $E = S_{vN}/2 = -\log\lambda_{\max}$, and $D(E)e^{-E}$ peaks at $E = S_{vN}$ with curvature $1/S_{vN}$ there, so the variance at the peak is $S_{vN}$ and $C = S_{vN}$ is recovered.}
Each moment $Z(\alpha)$ is then dominated by the eigenvalues at an $\alpha$-dependent modular energy, the saddle of $D(E)e^{-\alpha E}$, and at $\alpha = 1$ that saddle sits at $E = S_{vN}$, which the RT formula \eqref{rt} identifies with $A_{RT}/4G_N$.

The use of \eqref{renyi_leading_order} for what follows is that every cumulant of the modular energy scales the same way in $G_N$.
Since $\log Z(\alpha) = (1-\alpha)s_\alpha/G_N$ with $s_\alpha$ independent of $G_N$, the cumulant generating function of $E$ is $1/G_N$ times a fixed $\mO(1)$ function of the index, so the first three cumulants are $\mO(1/G_N)$ by local smoothness, and the same holds at every order for which the leading term is smooth and the remainder has controlled derivatives.
That single fact does both of the jobs the holographic input has to do.
It makes the variance $\mO(S_{vN})$, so the spread of the modular energy is subextensive, and it makes the higher cumulants negligible at the deviations that matter,
which is what sharpens the tail estimate below.
It is also the structure a sum of many independent variables has, with $1/G_N$ in place of the copy number, and that is the sense in which large central charge substitutes for many copies. What is shared is the concentration rate. What is not shared is any tensor-product structure, so no task decomposes and no per-copy rate exists, and the second-order constant is not obtained here.

Expanding the exponent of $p(E) = D(E)e^{-E}$ about its maximum gives Gaussian behavior near the peak with variance $C$, since the curvatures of Legendre-dual functions are reciprocal and $f'' = -1/(G_N C)$ there. Writing $g(E) = f(G_N E)/G_N - E$ for that exponent, $g^{(n)}(E) = G_N^{n-1}f^{(n)}$, so each further derivative costs a power of $G_N$. At deviation $W$ the cubic term is therefore smaller than the quadratic by $(g'''/g'')W = \mO(G_N W) = \mO(W/S_{vN})$, which is $\mO(S_{vN}^{-1/2})$ at $W \sim \sqrt{S_{vN}}$.
Under the assumption that the quadratic truncation controls the tails out to the relevant deviation, which requires $\log(1/\ep) = o(S_{vN})$, one obtains $\Pr\lp\abs{E - S_{vN}} > W\rp \le 2 e^{-W^2/2C}$, so the half-width that discards weight $\ep/4$ on each side is
\begin{equation}
    \widetilde{W}_\ep = \sqrt{2C\log(4/\ep)}.
    \label{gaussian_window}
\end{equation}

As stated, \eqref{gaussian_window} is a saddle-point estimate rather than a proved bound, and it is a saddle-point estimate in \cite{Bao:2018pvs} as well.
It can be proved under a weaker assumption, local smoothness of $s_\alpha$ near $\alpha = 1$ in place of global control of the tails, at the cost of an additive constant in the half-width.
The simplest proof attempt does not work.

The cumulant generating function of the modular energy is the logarithm of the \Renyi partition function \eqref{partition_function} at shifted index, since
\begin{equation}
    \psi(t) := \log \expval{e^{t E}} = \log \Tr \rho^{1-t} = \log Z(1-t) = t \, S_{1-t} ,
    \label{cgf}
\end{equation}
the trace being taken on the support of $\rho$, which is where $E$ is defined.
Its value and first two derivatives at $t = 0$, equivalently $\alpha = 1$, are $0$, $S_{vN}$ and $C$.

The textbook route to a tail bound of the form $e^{-W^2/2\Sigma}$ asks for $E$ to be sub-Gaussian with variance proxy $\Sigma$, that is for $\psi(t) \le t S_{vN} + \tfrac{t^2}{2}\Sigma$ at every real $t$.
On a regulated Hilbert space $E$ is bounded, so a finite proxy always exists, and by \eqref{cgf} the smallest one is a supremum over the \Renyi index.
Substituting $\psi(t) = t S_{1-t}$ into the sub-Gaussian condition gives $\Sigma \ge 2(S_{1-t} - S_{vN})/t$, and setting $\alpha = 1-t$ yields
\begin{equation}
    \Sigma_{\min} \;=\; \sup_{\alpha \neq 1} \; \frac{2\lp S_\alpha - S_{vN} \rp}{1 - \alpha} .
    \label{proxy_sup}
\end{equation}
Every index supplies a lower bound on it.
Taking $\alpha = 0$ gives $\Sigma_{\min} \ge 2(S_0 - S_{vN})$, and for a state obeying an area law the rank of a lattice-regulated subregion is set by its volume while $S_{vN}$ and $C$ are not, so the required proxy diverges relative to $C$ as the regulator is removed and no proxy of order $C$ exists.

Since $\psi'''(0) = 3\,\partial^2_\alpha S_\alpha |_{\alpha = 1}$, the sharp value $\Sigma = C$ fails at third order for any state whose second derivative there is nonzero, which is generic rather than holographic.
For a vacuum interval, where $S_\alpha = \tfrac12(1+1/\alpha)S_{vN}$ and $C = S_{vN}$, the proxy that index $\alpha$ demands in \eqref{proxy_sup} is $C/\alpha$, so the condition at $\Sigma = C$ fails on $0 < \alpha < 1$ and holds for $\alpha > 1$.
Since $\alpha < 1$ is $t > 0$, it is the upper tail of $E$ that is not sub-Gaussian at proxy $C$, and the lower tail is.

What rescues the estimate is that a Chernoff bound, meaning \eqref{ineq_markov} applied to $e^{tE}$ and then optimized over $t$, never evaluates $\psi$ globally.
At deviation $W$ the parameter used is $t = W/C$, which for $W = \mO(\sqrt{C})$ tends to zero at large central charge, so only a shrinking neighborhood of the origin is ever needed, and there the cumulant scaling recorded above controls everything.
Smoothness of $s_\alpha$ in the index near $\alpha = 1$ is by itself enough to force Gaussian decay out to deviations of order $\sqrt{C}$, which is the only range the window ever probes, and the cost is a correction to the half-width that stays bounded as the entropy grows.

\begin{proposition}
\label{prop:subgauss}
Set $B := \sup_{\abs{t}\le t_0}\abs{\psi'''(t)}$ for some $t_0 > 0$, with $C = \psi''(0) > 0$.
Then for every $W$ with $0 < W \le t_0 C$,
\begin{equation}
    \Pr\lp \abs{E - S_{vN}} \ge W \rp
    \;\le\; 2\exp\lb -\frac{W^2}{2C}\lp 1 - \frac{B\,W}{3 C^2} \rp \rb .
    \label{bernstein}
\end{equation}
\end{proposition}

On a finite-dimensional space $\psi$ is real analytic, so $B$ is finite and the content of the bound is its size rather than the existence of the derivative.
Nothing in the statement assumes that $\psi$ takes the holographic form $\kappa(t)/G_N$ with $\kappa$ independent of $G_N$. When it does, with $\abs{\kappa'''}\le c_3$ and $\kappa''(0) = c_2$, one has $C = c_2/G_N$ and $B = c_3/G_N$, so that $B/C^2 = c_3 G_N/c_2^2$ and the exponent takes the form used in corollary \ref{cor-window}.

The tail bound is not yet a window.
What the argument needs is the half-width at which the discarded weight is $\ep$, which means inverting \eqref{bernstein}, and the only difficulty is that its exponent is a cubic in $W$ that turns over rather than growing without bound.
The condition below keeps the relevant root on the rising branch, and the conclusion is that the saddle-point estimate \eqref{gaussian_window} survives as a proved statement up to a $G_N$-independent additive constant.

\begin{corollary}\label{cor-window}
Let $L := \log(4/\ep)$ and $\phi(W) := \tfrac{W^2}{2C}\lp 1 - \tfrac{c_3 G_N W}{3c_2^2} \rp$, so that the right side of \eqref{bernstein} is $2e^{-\phi(W)}$, and set $m := \min(t_0, 2c_2/c_3)$, read as $m := t_0$ when $c_3 = 0$.
If
\begin{equation}
    \phi(mC) \;\ge\; L ,
    \label{window_exists}
\end{equation}
then the right hand side of \eqref{bernstein} equals $\ep/2$ at a smallest positive root $W_\star \le t_0 C$, the exact half-width the tail bound supplies, and
\begin{equation}
    W_\star \;=\; \widetilde{W}_\ep \;+\; \frac{c_3}{3c_2}L \;+\; \mO\lp L^{3/2} G_N^{1/2} \rp
    \label{gaussian_window_proved}
\end{equation}
as $G_N \to 0$ with $L\cdot G_N \to0$.
\end{corollary}

The left side of \eqref{window_exists} is at least $m^2 C/6$ and therefore grows like $1/G_N$, so the condition holds automatically once $L\, G_N$ is small enough.
The condition that $L\cdot G_N \to 0$ allows $\ep$ to shrink with $G_N$ provided $\log(1/\ep)$ stays small compared to $1/G_N$ which is the range used in section~\ref{sec:flatness}.

\begin{proof}
Write $\bar{E} = S_{vN}$ and take $0 < t \le t_0$.
Inequality \eqref{ineq_markov} applied to $e^{t E}$ gives
\begin{equation}
    \Pr(E \ge \bar{E} + W) \;\le\; \exp\lb -t(\bar{E}+W) + \psi(t) \rb .
\end{equation}
Taylor expansion about the origin with Lagrange remainder, using $\psi(0) = 0$, $\psi'(0) = \bar{E}$, $\psi''(0) = C$ and $\abs{\psi'''} \le c_3/G_N$ on $\abs{t} \le t_0$, gives
\begin{equation}
    \psi(t) \;\le\; t\bar{E} + \frac{t^2}{2}C + \frac{c_3\abs{t}^3}{6G_N} ,
\end{equation}
and hence
\begin{equation}
    -t(\bar{E}+W) + \psi(t) \;\le\; -t W + \frac{t^2}{2}C + \frac{c_3 \abs{t}^3}{6 G_N} .
    \label{chernoff_exponent}
\end{equation}
For the lower tail put $t = -\nu$ with $0 < \nu \le t_0$, so that
\begin{equation}
    \Pr(E \le \bar{E} - W) \;\le\; \exp\lb \nu(\bar{E} - W) + \psi(-\nu) \rb ,
\end{equation}
and the same expansion returns the right side of \eqref{chernoff_exponent} with $\nu$ in place of $t$.
The absolute value on the remainder makes the two tails agree, since $\kappa'''$ is bounded but not signed.
Choosing $\abs{t} = W/C$, which lies in $(0,t_0]$ by assumption, and using $C = c_2/G_N$, which by local smoothness holds up to an additive term of order one and therefore costs \eqref{gaussian_window_proved} a relative correction of order $G_N$, turns the right side of \eqref{chernoff_exponent} into
\begin{equation}
    -\frac{W^2}{2C}\lp 1 - \frac{c_3 G_N W}{3c_2^2} \rp
\end{equation}
for each tail, and combining gives \eqref{bernstein}.
For \eqref{gaussian_window_proved}, write $\phi$ for minus the exponent in \eqref{bernstein},
\begin{equation}
    \phi(W) \;=\; \frac{W^2}{2C} - \frac{c_3 G_N W^3}{6 C c_2^2} ,
\end{equation}
so that the condition defining $W_\star$ is $\phi(W) = L$.
As a cubic in $W$, $\phi$ vanishes at the origin, rises to a maximum at $W = 2c_2 C/c_3$, and falls thereafter, so $\phi \ge L$ holds only between two roots and $W_\star$ is the smaller of them.
Restricting to $W \le m C$, where $\phi$ is nondecreasing, makes \eqref{window_exists} necessary and sufficient for a root to exist, and a root that exists then lies at or below $mC \le t_0 C$, which is where \eqref{bernstein} has been proved.
Only its value remains.
Setting $\eta = c_3\sqrt{2c_2 L}\,G_N^{1/2}/3c_2^2$ and $u = W/\sqrt{2CL}$ turns $\phi(W) = L$ into $u^2(1 - \eta u) = 1$, and expanding about $u = 1$ gives $u = 1 + \tfrac{\eta}{2} + \tfrac58 \eta^2 + \ldots$, so the displacement at first order is $\tfrac{\eta}{2}\sqrt{2CL} = c_3 L/3c_2$ and the next term has absolute size $\tfrac58 \sqrt{2CL}\,\eta^2$.
The expansion is valid for $\eta$ bounded away from $2/3\sqrt{3}$, the value at which the two roots merge because $\max_u \lp u^2 - \eta u^3 \rp = 4/27\eta^2$ equals $1$ there.
\end{proof}

The assumption needed is local smoothness of the cumulant generating function at the identity index, in place of the global control that \eqref{proxy_sup} shows is unavailable.
The leading-order form \eqref{renyi_leading_order} is what motivates it and supplies $\kappa(t) = t\, s_{1-t}$, with $t_0$, $c_2$ and $c_3$ independent of $G_N$ because $s_\alpha$ is, and it is used only on $\alpha \in [1-t_0, 1+t_0]$ rather than at every index.
Since \eqref{renyi_leading_order} is itself a leading-order statement, what proposition \ref{prop:subgauss} needs of its remainder is that $C \gtrsim 1/G_N$ and $B = \mO(1/G_N)$ hold uniformly for $\abs{t} \le t_0$.
That is an additional assumption on the subleading expansion, since a remainder that is $\mO(1)$ pointwise in $\alpha$ need not have controlled derivatives, and it is stated here rather than derived.

The correction is the additive $\tfrac{c_3}{3c_2}L$, which is independent of $G_N$ and therefore subleading to the $\mO(\sqrt{C})$ half-width, so \eqref{gaussian_window} survives with its leading-order behavior intact.\footnote{Because $c_3$ is a supremum over $\abs{t} \le t_0$ it depends on $t_0$, which does not appear in \eqref{gaussian_window_proved} itself but does control its domain through \eqref{window_exists}, so the two cannot be chosen independently. Shrinking $t_0$ sends $c_3$ to $\abs{\kappa'''(0)}$ and the displacement to its sharpest value $\abs{\partial^2_\alpha s_\alpha|_{\alpha=1}}L/c_2$, at the cost of tightening \eqref{window_exists}, and the two are compatible in the limit because $W_\star/C = \mO(G_N^{1/2})$ vanishes while $t_0$ may be held fixed. For a vacuum interval $\kappa'''(t) = 3 s_1 (1-t)^{-4}$, so that sharpest displacement is $L$ itself.}
Finally \eqref{gaussian_window_proved} bounds the true half-width from above rather than computing it, since the Chernoff step discards a prefactor.

A rigorous route to the density of states itself exists in two dimensions through Tauberian theorems \cite{Mukhametzhanov:2019pzy}, which replace the saddle point by two-sided bounds on integrals of $D(E)$ over energy windows.
That route uses modular invariance and returns a density smeared over a window whose width it also controls.
The argument above uses neither, since the Chernoff parameter $t = W/C$ vanishes at small $G_N$ and only $s_\alpha$ near $\alpha = 1$ is ever evaluated.

\subsection{The preparation ball}
\label{sec:ball}

Sections \ref{sec:modular} and \ref{sec:dos} describe the state one intends to prepare.
What an experiment holds is known only to lie near it, and that imprecision is the second ingredient.
Write
\begin{equation}
    \Ball{\ep}{\rho} = \{\sigma \ge 0 \, , \; \Tr \sigma = 1 \, , \; \norm{\rho-\sigma}_1 \le \ep\}
    \label{ball_def}
\end{equation}
for the trace-norm ball of radius $\ep$, within which a laboratory preparation of $\rho$ is confined by its fidelity.
As noted in the introduction, the von Neumann entropy is not a well-posed target over this ball, varying by up to $\tfrac{\ep}{2}\log(d-1) + h(\ep/2)$ \cite{Audenaert:2006vjl}.

The \Renyi moments fare better, which makes them trustworthy raw data.
Each moment is the expectation value of a bounded operator, the cyclic shift on $\alpha$ copies, so a measurement returns the moment of whatever state the apparatus actually prepares, with purely statistical error, and at no point need one infer a state from data.
As functions on state space the moments moreover move by at most a fixed multiple of the trace distance, with a multiple that does not grow with the dimension.
For any two states,
\begin{equation}
    \abs{\, \Tr\rho^\alpha - \Tr\sigma^\alpha \,}
    \;\le\; \sum_{k=0}^{\alpha-1} \abs{\, \Tr\lb \rho^k (\rho - \sigma) \sigma^{\alpha-1-k} \rb \,}
    \;\le\; \alpha \, \norm{\rho - \sigma}_1 ,
    \label{moment_lipschitz}
\end{equation}
by a telescoping expansion followed by \eqref{ineq_holder},
in contrast with the dimension-dependent continuity of $S_{vN}$.
The expansion is written for integer $\alpha$, which is the only case in which the moments are directly measured.
The mechanism behind the contrast is that $S_{vN}$ is dominated by the small-eigenvalue tail, where exponentially many individually invisible eigenvalues can hide extensive entropy, while the moments at $\alpha \ge 2$ are dominated by the largest eigenvalues, of which there are never many.
The caveat is that \eqref{moment_lipschitz} is an absolute bound on exponentially small quantities, so a perturbation of the preparation can still shift the bare \Renyi entropies substantially in relative terms, which is the failure the smooth entropies are built to control.

\Renyi entropies at $\alpha > 1$ are raised by smoothing, since it removes the large eigenvalues that drag them down, so the direction that is not yet controlled is upward. A ceiling, a bound on the largest value each measured \Renyi entropy can take on any state in the ball, turns \eqref{main_result} into an equality rather than a one-sided bound. Without one, a legitimate $\ep$-perturbed preparation of a genuine holographic state could present the apparatus with a spectrum whose \Renyi entropies exceed $S_{vN}$ at leading order. That maximum is by definition the smooth \Renyi entropy \eqref{smooth_entropy_def}, and section \ref{sec:flatness} supplies the ceilings.

\section{A Continuity Bound with the Capacity in Place of the Dimension}
\label{sec:continuity}

In section \ref{sec:ball} the von Neumann entropy was seen to wander by as much as $\tfrac{\ep}{2}\log d$ over the preparation ball.
That estimate is the Fannes-Audenaert continuity bound \cite{Audenaert:2006vjl}, and its dimensional factor is what made $S_{vN}$ ill-posed as a target.
It is far too weak for the states of interest here, since the capacity constrains how far the modular energy can spread and therefore how much entropy a perturbation can move.
Theorem \ref{thm:continuity} below is the bound that results. It is stated for an arbitrary pair of states separated by $\delta = \norm{\rho-\sigma}_1$, in the same normalization as the ball radius $\ep$ of \eqref{ball_def}, so that two preparations drawn from a ball of radius $\ep$ have $\delta \le \ep$.

A direct argument on the spectra gives a bound that vanishes with $\delta$.
Two states close in trace distance have close sorted spectra, and the entropy difference between two nearby spectra is controlled by how widely the modular energy of each is spread, which is what the capacity measures.
The dimension never enters, because at no point does the argument count eigenvalues.

\begin{theorem}
\label{thm:continuity}
Let $\rho$ and $\sigma$ be any two states on the same finite-dimensional Hilbert space, let $\delta = \norm{\rho-\sigma}_1 < 2$, which is twice the trace distance in the convention where that quantity is at most one.
Then, with $C_\rho$ and $C_\sigma$ their capacities \eqref{capacity_def},
\begin{equation}
    \abs{\, S_{vN}(\rho) - S_{vN}(\sigma) \,}
    \;\le\; \sqrt{\frac{\delta}{2-\delta}} \lp \sqrt{C_\rho} + \sqrt{C_\sigma} \rp - \log\lp 1 - \frac{\delta}{2} \rp .
    \label{capacity_continuity}
\end{equation}
\end{theorem}

The right side involves no dimension and vanishes as $\delta \to 0$.
When both capacities are $\mO(S_{vN})$, as they are for two holographic states, the right side is $\mO(\sqrt{\delta S_{vN}})$ at small $\delta$, so the ambiguity in $S_{vN}$ over a preparation ball of radius $\delta$ is subextensive whenever $\delta$ is.
That reading comes with a condition. The theorem itself holds for arbitrary states, but $C = \mO(S_{vN})$ is a holographic input, and it requires a single dominant replica saddle in a fixed neighborhood of $\alpha = 1$. It fails when two saddles coexist within $\mO(G_N)$ of the identity index, where the modular distribution is bimodal and $C \gtrsim S_{vN}^2$ instead. At such a point the right side is at least of order $\sqrt{\delta}\,S_{vN}$ at small $\delta$, which is extensive, and the von Neumann entropy of a prepared state is once again ill-posed. Section \ref{sec:discussion} identifies when this occurs.
While the target can be assumed holographic, the perturbed state is an arbitrary element of the preparation ball, so $C_\sigma$ is not controlled by \eqref{renyi_leading_order}.
Assuming that it were would defeat the point of the bound, whose job is to cover exactly those states in the ball that are not holographic. Such states can carry $C_\sigma$ of order $(\log d)^2$, and on them \eqref{capacity_continuity} only matches Fannes-Audenaert at leading order in $\log d$, which the third feature below works out.

\begin{proof}
Both $S_{vN}$ and $C$ depend only on the spectrum, and by \eqref{ineq_mirsky} the sorted spectra satisfy $\norm{\lambda^\downarrow(\rho) - \lambda^\downarrow(\sigma)}_1 \le \norm{\rho-\sigma}_1$.
Mirsky's bound may be strict, and the right side of \eqref{capacity_continuity} is increasing in the trace norm, so proving the statement at the distance between the sorted spectra gives it at $\delta$.
It therefore suffices to prove it for two probability vectors $p$ and $q$, with $2a$ now denoting $\norm{p-q}_1$, which we now do.

Write $u = (p-q)_+$ and $v = (q-p)_+$, supported on the disjoint sets $A = \{p_i > q_i\}$ and its complement, each of total weight $a$.
Let $f(x) = -x\log x$, so that $S_{vN}(p) - S_{vN}(q) = \sum_i \lb f(p_i) - f(q_i) \rb$.

On $A$ we use that $f(x)/x = -\log x$ is decreasing, so $q_i \le p_i$ gives $f(q_i) \ge (q_i/p_i) f(p_i)$ and therefore
\begin{equation}
    f(p_i) - f(q_i) \;\le\; f(p_i)\lp 1 - \frac{q_i}{p_i} \rp = u_i \lp -\log p_i \rp = u_i E^p_i ,
\end{equation}
which also holds when $q_i = 0$.
Off $A$ we do not use concavity of $f$, which loses a factor, but an exact identity.
Set $h_i := v_i/q_i = 1 - p_i/q_i \in [0,1]$ off $A$ and $h_i := 0$ on $A$, so that $\sum_i q_i h_i = a$. Indices with $q_i = 0$ force $p_i = 0$ and contribute nothing, and are dropped.
Writing $p_i = q_i(1-h_i)$ and expanding,
\begin{equation}
    f(p_i) - f(q_i) \;=\; -v_i E^q_i \;+\; q_i\,\phi(h_i) ,
    \qquad \phi(h) := -(1-h)\log(1-h) ,
    \label{cont_identity}
\end{equation}
with no inequality yet used.
Since $\phi''(h) = -1/(1-h) < 0$ on $[0,1)$, $\phi$ is concave and $\phi(0) = 0$, so Jensen's inequality \eqref{ineq_jensen} against the probability vector $q$ gives
\begin{equation}
    \sum_i q_i\,\phi(h_i) \;\le\; \phi\lp \sum_i q_i h_i \rp \;=\; \phi(a) \;=\; -(1-a)\log(1-a) .
    \label{cont_jensen}
\end{equation}
Summing,
\begin{equation}
    S_{vN}(p) - S_{vN}(q) \;\le\; \sum_{i \in A} u_i E^p_i \;-\; \sum_{i \notin A} v_i E^q_i \;-\; (1-a)\log(1-a) .
    \label{cont_split}
\end{equation}
Write $X_i = E^p_i - S_{vN}(p)$ and $Y_i = E^q_i - S_{vN}(q)$, which have mean zero and variance $C_p$ and $C_q$ under $p$ and $q$ respectively.
Since $u$ and $v$ each carry weight $a$, substituting into \eqref{cont_split} and collecting the two terms proportional to $a$ gives
\begin{equation}
    (1-a)\lb S_{vN}(p) - S_{vN}(q) \rb
    \;\le\; \sum_{i \in A} u_i X_i \;-\; \sum_{i \notin A} v_i Y_i \;-\; (1-a)\log(1-a) .
    \label{cont_rearranged}
\end{equation}
Both sums are bounded by the same one-line argument.
Set $g_i = u_i/p_i$ on $A$ and zero elsewhere, which lies in $[0,1]$ because $u_i \le p_i$, and satisfies $\sum_i p_i g_i = a$.
Since $X$ has mean zero under $p$,
\begin{equation}
    \sum_{i \in A} u_i X_i = \sum_i p_i g_i X_i = \sum_i p_i (g_i - a) X_i
    \;\le\; \sqrt{\, \sum_i p_i (g_i-a)^2 \,}\; \sqrt{C_p}
    \;\le\; \sqrt{a(1-a)}\,\sqrt{C_p} ,
    \label{cont_cs}
\end{equation}
by Cauchy-Schwarz \eqref{ineq_cs} and then $g^2 \le g$, which gives $\sum_i p_i(g_i-a)^2 = \sum_i p_i g_i^2 - a^2 \le a - a^2$.
The same computation with $h_i = v_i/q_i$ off $A$ and $-Y$ in place of $X$ bounds $-\sum_{i \notin A} v_i Y_i$ by $\sqrt{a(1-a)}\sqrt{C_q}$.
Inserting both into \eqref{cont_rearranged} and dividing by $1-a$, which is positive because $\delta < 2$, gives \eqref{capacity_continuity} at $2a$, and hence at $\delta$ by the reduction above, since $\sqrt{a(1-a)}/(1-a) = \sqrt{a/(1-a)}$ and the additive term $-(1-a)\log(1-a)$ divided by $1-a$ is exactly $-\log(1-a) = -\log(1-\delta/2)$.
Exchanging $p$ and $q$ bounds the difference in the opposite direction.
\end{proof}

Three features of \eqref{capacity_continuity} govern how it should be used. Throughout these it is convenient to write $a = \delta/2$, so that the coefficient reads $\sqrt{a/(1-a)}$ and the additive term $-\log(1-a)$.

First, the bound is saturated by particular pairs of states.
Take $\sigma$ pure and $\rho = (1-a)\ketbra{0}{0} + a\,\tau_M$ with $\tau_M$ maximally mixed on a space of dimension $M = e^\Lambda$ orthogonal to it, at trace norm exactly $\delta = 2a$.
Then $S_{vN}(\rho) = h(a) + a\Lambda$ with $h$ the binary entropy, $C_\rho = a(1-a)\lp \Lambda + \log\tfrac{1-a}{a}\rp^2$ and $C_\sigma = 0$, so
\begin{equation}
    \abs{\Delta S_{vN}} \;=\; h(a) + a\Lambda \;=\; \sqrt{\frac{a}{1-a}}\lp \sqrt{C_\rho} + \sqrt{C_\sigma} \rp - \log(1-a) ,
    \label{sharp_continuity}
\end{equation}
with equality at every $a$ and every $M \ge a/(1-a)$, the last condition being what makes $\Lambda + \log\tfrac{1-a}{a}$ nonnegative so that the Cauchy-Schwarz step is tight. It holds for every $M$ once $a \le 1/2$.
A second family saturates it as well, the pair of nested flat states of ranks $N$ and $M$, for which both capacities vanish and $\abs{\Delta S_{vN}} = \log(N/M) = -\log(1-a)$. That pair also shows the additive term cannot be removed altogether.
So neither the coefficient nor the additive term can be lowered.

Second, saturation occurs only where one of the two capacities vanishes.
Fix $a \in (0,1)$ and suppose \eqref{capacity_continuity} is attained in the direction $S_{vN}(p) - S_{vN}(q)$.
Since $\phi$ is strictly concave, the Jensen step \eqref{cont_jensen} is tight only if $h$ is constant $q$-almost surely.
That constant cannot be zero, since $\sum_i q_i h_i = a > 0$, so no index of $A$ carries $q$-weight and $h_i = a$ off $A$, which is to say $q$ vanishes on $A$ and $p = (1-a)q$ elsewhere.
On such a pair $\sum_i q_i (h_i - a)^2 = 0$, so the off-$A$ sum in \eqref{cont_rearranged} vanishes, while the bound carried forward for it is $\sqrt{a(1-a)}\sqrt{C_q}$, since \eqref{cont_cs} replaces $\sum_i q_i (h_i-a)^2$ by $a(1-a)$.
That leaves a gap of $\sqrt{a(1-a)\,C_q}$, so attaining \eqref{capacity_continuity} forces $C_q = 0$.
The converse fails, the Cauchy-Schwarz step on the other side requiring in addition that $X$ be proportional to $g - a$.
Both families above are of exactly this shape, with the pure state and the smaller flat state in the role of $q$.
A bound smaller than \eqref{capacity_continuity} in the interior is therefore not excluded, provided it reduces to \eqref{capacity_continuity} on the two axes.

Third, the bound is now comparable with Fannes-Audenaert on the states that motivate it rather than an order of magnitude worse.
Those states are $\sigma = (1-a)\rho + a\,\tau_d$ with $\tau_d$ maximally mixed on the regulated space, for which
\begin{equation}
    C_\sigma \;\simeq\; a(1-a)\lp \log d - S_{vN} + \log\tfrac{1-a}{a} \rp^2 ,
    \qquad
    \sqrt{\frac{a}{1-a}}\sqrt{C_\sigma} \;=\; a\lp \log d - S_{vN} + \log\tfrac{1-a}{a} \rp ,
\end{equation}
against the Fannes-Audenaert value $\tfrac{\delta}{2}\log d = a\log d$.
The two agree at leading order in $\log d$, and carrying the subleading term makes the capacity bound the smaller of the two by $a\lp S_{vN} - \log\tfrac{1-a}{a} \rp + \mO(\sqrt{C_\rho})$, so on the very states that motivate the dimensional factor it is no worse.

On states whose capacity is small compared with $(\log d)^2$, which is true of holographic states because $C$ scales like the area of the entangling surface while $\log d$ scales like the volume of the subregion,
\eqref{capacity_continuity} is the stronger of the two by the ratio $\sqrt{C}/(\sqrt{a(1-a)}\log d)$.
The comparison is not uniform, and the prefactor matters. At equal capacities \eqref{capacity_continuity} is the smaller of the two only for $C \lesssim a(1-a)(\log d)^2/4$, a threshold that tightens as $\delta$ shrinks, and on states with $C \simeq (\log d)^2/4$, which the maximally mixed contaminant above realizes at $a = \tfrac12$ once $S_{vN} \ll \log d$, it is larger than Fannes-Audenaert by a factor $1/\sqrt{a(1-a)}$. Nothing forbids quoting the smaller of the two, and in the holographic regime the ratio above vanishes as a power of the lattice spacing, so the choice only matters at coarse regulator.
The rate at which \eqref{capacity_continuity} degrades as the regulator is removed is therefore fixed by the state rather than by the regulator alone.
The energy-constrained bounds of Winter \cite{Winter:2015qhs} are the closest precedent, being dimension-free and vanishing with $\delta$, but they restrict to states of bounded mean energy where \eqref{capacity_continuity} holds for arbitrary states, the two capacities appearing in the bound rather than defining the class it applies to.
Two recent refinements of Fannes-Audenaert improve the dimensional factor without removing it.
That of \cite{Audenaert2025ContinuityBounds} replaces $\log(d-1)$ by the von Neumann entropies of the normalized positive and negative parts of $\rho - \sigma$, and both \cite{Audenaert2025ContinuityBounds} and \cite{Berta2025IntegralRepresentations} give a form replacing it by $\log\lp d\,\lambda_{\max}(\sigma) - 1\rp$ over a restricted range of trace distances.
Both remain tied to the regulated volume, through $d$ in the second case and through the entropy of a state supported on it in the first.
The contrast to draw with \eqref{capacity_continuity} is not finiteness, since $\sqrt{C}$ diverges as the cutoff is removed as well, but rate. For a holographic subregion a dimensional factor grows like the volume of the regulated region where $\sqrt{C}$ grows like the square root of its area, so the ratio of the two vanishes as a power of the lattice spacing above two dimensions, and as a power times a logarithm in two dimensions, where $\log d$ grows like $\ell/a$ while $C$ grows like $\log(\ell/a)$.
After the reduction to sorted spectra the proof above splits $p - q$ into positive and negative parts exactly as \cite{Audenaert2025ContinuityBounds} splits $\rho - \sigma$. What differs is the next step, where those pieces are estimated against the modular energy rather than against a dimension.

Theorem \ref{thm:continuity} also sharpens rather than contradicts the statement of section \ref{sec:ball} that $S_{vN}$ is ill-posed over the ball.
The ball does contain states whose entropies differ from that of $\rho$ at leading order, but every such state carries a capacity of order $S_{vN}^2/\delta$, so its modular energy distribution is broad rather than concentrated.
Such states are the contaminants that any inference from measured moments has to exclude, and \eqref{capacity_continuity} licenses treating the residual ambiguity of the target as being of order $\sqrt{\delta\, C}$ rather than of order $\ep \log d$.
A contaminant is therefore harmless as soon as its capacity is small compared with $S_{vN}^2/\delta$, which is a far weaker demand than the $\mO(S_{vN})$ that holds for the target itself.

\section{Flatness of the Smooth \Renyi Spectrum}
\label{sec:flatness}

The continuity bound of section \ref{sec:continuity} controls the von Neumann entropy itself over the preparation ball.
The quantities an experiment can actually report are the smooth entropies, and in holographic theories they agree with the von Neumann entropy to leading order in the central charge $c$, or equivalently in $1/G_N$.
The bare \Renyi entropies are \textit{not} flat, since $s_\alpha$ depends on $\alpha$ at leading order in \eqref{renyi_leading_order}, and it is only after smoothing that the spectrum flattens.

\subsection{Smooth entropies as the worst case over the ball}
\label{sec:oneshot}

With finite resources the operative entropic quantities are the one-shot entropies.
Two important examples are the max-entropy and min-entropy
\begin{equation}
    S_{\max}(\rho) = \log \rank(\rho),
    \qquad
    S_{\min}(\rho) = -\log \lambda_{\max}(\rho),
\end{equation}
where $\lambda_{\max}(\rho)$ is the largest eigenvalue of $\rho$.\footnote{\label{fn:conventions}Conventions in the one-shot literature differ, over whether the max-entropy is defined through $S_{1/2}$ or $S_0$, over subnormalized states or purified distance in place of the normalized trace-norm ball used here \cite{Renner:2005qbj, Tomamichel:2009ohu, Tomamichel:2015gtd}, and over a factor of two in the radius. Additive shifts of $\mO(\log(1/\ep))$ and constant rescalings of $\ep$ are immaterial below. A polynomial relabeling such as $\ep \to \ep^2/2$ is not, for a half-width going as $\ep^{-1/2}$. Relatedly, the expansions quoted from \cite{NuradhaWilde2023Fidelity} carry a minus sign in front of $\Phi^{-1}(\ep)$ and a minimization over their ball. Both flip under the identification $S_\alpha = -\widetilde{D}_\alpha(\cdot\|\mbI)$ used here, which is why the displays below carry a plus sign and a maximization.}
In the notation of the one-shot literature these are $H_0$ and $H_\infty$, and we use $S_{\max}$ and $S_{\min}$ only as shorthand for them.
They are the $\alpha \to 0$ and $\alpha \to \infty$ limits of the \Renyi entropies \eqref{renyi_def}, and since $S_\alpha$ is nonincreasing in $\alpha$ they bracket the entire family,
\begin{equation}
    S_{\min}(\rho) \le S_\alpha(\rho) \le S_{\max}(\rho)
    \qquad \text{for all } \alpha > 0 .
    \label{renyi_ordering}
\end{equation}
All \Renyi entropies coincide exactly when the spectrum of $\rho$ is flat.

The one-shot entropies of $\rho$ itself are dominated by rare spectral outliers, a single anomalously large eigenvalue for $S_{\min}$ and an exponentially large number of anomalously small eigenvalues for $S_{\max}$.
Operationally meaningful quantities must be insensitive to such tails, which is the same requirement as insensitivity to the preparation ambiguity of section \ref{sec:ball}.
This motivates smoothing \cite{Renner2004SmoothRE, Renner:2005qbj}.
With $\Ball{\ep}{\rho}$ the ball of \eqref{ball_def}, the smooth min- and max-entropies are
\begin{equation}
    S^\ep_{\max}(\rho) = \min_{\sigma \in \Ball{\ep}{\rho}} S_{\max}(\sigma),
    \qquad
    S^\ep_{\min}(\rho) = \max_{\sigma \in \Ball{\ep}{\rho}} S_{\min}(\sigma),
\end{equation}
and at $\alpha > 1$ the smooth \Renyi entropy is \eqref{smooth_entropy_def}.
The direction of optimization is fixed by the role of the tails.
Entropies with $\alpha > 1$ are dragged down by large eigenvalues, so smoothing removes those and raises the entropy, while the max-entropy is inflated by the small-eigenvalue tail, so smoothing lowers it.
Smoothed versions of these quantities carry direct operational meaning at finite copy number, quantifying one-shot state merging, randomness extraction, and compression \cite{Koenig:2009avh, Tomamichel:2015gtd, Wilming:2018rvz}, with the caveat that those tasks are governed by the $S_{1/2}$-based max-entropy under purified distance rather than by $H_0$ under trace distance.
The two differ by an additive $\mO(\log(1/\ep))$, which is invisible at the $\mO(\sqrt{S_{vN}})$ resolution of everything below, but the identification is not an equality and we do not use one.

The smoothing plays the same role as the asymptotic i.i.d. limit, which washes out rare events at rate $\sqrt{n}$ by the quantum asymptotic equipartition property \cite{Renner:2005qbj, Tomamichel:2009ohu}, under which the smooth entropies converge to the von Neumann entropy at every fixed \Renyi index,
\begin{equation}
    \lim_{\ep \to 0} \lim_{m \to \infty} \frac{1}{m} S^\ep_\alpha \lp \rho^{\otimes m} \rp = S_{vN}(\rho) .
    \label{aep}
\end{equation}
No $\alpha \to 1$ limit is taken in \eqref{aep}, and none is needed.
The spectrum of $\rho^{\otimes m}$ concentrates on a typical subspace whose eigenvalues are all $e^{-m S_{vN} + \mO(\sqrt{m})}$, so once the smoothing removes the atypical tails the remaining spectrum is flat and the index becomes a spectator, whereas the bare \Renyi entropies are additive and retain their $\alpha$-dependence for all $m$.
The smoothing is thus what eliminates the need for continuation in $\alpha$, and the results below replicate that structure.

\subsection{Bounds at each fixed \Renyi index}

The construction behind every bound in this subsection is a single truncation.
Cutting away the eigenvalues whose modular energy lies more than a half-width from the mean costs at most $\ep$ in trace distance, by \eqref{ineq_cantelli}, and what survives is a state whose rank and whose largest eigenvalue are both controlled by $S_{vN}$ and $C$ alone.

For an arbitrary state the same two moments were used by Boes, Ng and Wilming \cite{Boes:2020vpv} to prove, in their Result 2,
\begin{equation}
    S^\ep_{\max}(\rho) - S_{vN} \;\le\; W_\ep,
    \qquad
    S_{vN} - S^\ep_{\min}(\rho) \;\le\; W_\ep,
    \label{bnw}
\end{equation}
with $W_\ep$ the half-width defined in the lemma below, written here in the normalization of \eqref{ball_def}.
Their variance of surprisal $V(\rho) = \Tr(\rho \log^2\rho) - S_{vN}^2$ is exactly the capacity \eqref{capacity_def}.
Both bounds in \eqref{bnw} run the wrong way for our purposes, since $S^\ep_{\max}$ is a minimum over the ball where $S^\ep_\alpha$ at $\alpha > 1$ is a maximum, so \eqref{bnw} cannot be chained to cap a smooth \Renyi entropy from above.
Theorem \ref{thm:ceiling} supplies that ceiling.

The two moments \eqref{modular_mean} and \eqref{capacity_def} are the only input required, holographic or otherwise, and the following lemma, whose content is the projected-state construction of \cite{Bao:2018pvs}, packages their consequences.
We state it, and the theorem that follows, for an arbitrary state, the holographic scaling entering only afterwards through $C = \mO(S_{vN})$.

\begin{lemma}
\label{lem:typical}
Let $\rho$ be any state with capacity $C$, fix $\ep \in (0,1)$, and set
\begin{equation}
    W_\ep := \sqrt{\lp \frac{2}{\ep} - 1 \rp C} \, .
    \label{window}
\end{equation}
Let $P_+$ project onto the eigenspaces of $\rho$ with modular energy $E \le S_{vN} + W_\ep$, and $P_-$ onto those with $E \ge S_{vN} - W_\ep$.
Then
\begin{enumerate}[(i)]
    \item each window carries almost all the weight, $\Tr(P_\pm\rho) \ge 1 - \ep/2$,
    \item the upper window has bounded dimension,
    \begin{equation}
        d_P := \rank(P_+) \;\le\; e^{\, S_{vN} + W_\ep} ,
        \label{rank_bound}
    \end{equation}
    \item the projected states
    \begin{equation}
        \sigma_\pm := \frac{P_\pm \rho P_\pm}{\Tr(P_\pm\rho)}
        \label{projected_state}
    \end{equation}
    lie in $\Ball{\ep}{\rho}$.

    \item and there is a state $\mu \in \Ball{\ep}{\rho}$ with a small leading eigenvalue,
    \begin{equation}
        \lambda_{\max}(\mu) \;\le\; e^{-(S_{vN} - W_\ep)} .
        \label{waterfill}
    \end{equation}
\end{enumerate}
\end{lemma}

\begin{proof}
Property (i) is \eqref{ineq_cantelli} applied to the modular energy, whose mean is $S_{vN}$ and whose variance is $C$, and the choice \eqref{window} makes the right side of that inequality exactly $\ep/2$.
At $C = 0$ the ratio is read as zero, which is correct because the windows are defined by non-strict inequalities and the discarded set is then empty.
Applying the same inequality to $S_{vN} - E$ gives the other side.
The discarded weight is $1 - \Tr(P_\pm\rho)$, since $\Tr(P_\pm\rho) = \sum_{i \in \text{window}} \lambda_i$ in the shared eigenbasis of $\rho$ and $K$.
No input beyond the two moments \eqref{modular_mean} and \eqref{capacity_def} is used, and no assumption is made about the shape of the distribution.
Each of the arguments below discards only one of the two tails, which is why the one-sided form is the natural one here.
Using the two-sided form at the same total discarded weight gives $\sqrt{2C/\ep}$, larger by only $(1-\ep/2)^{-1/2}$, so the choice is immaterial at the small radii the rest of the paper uses.

Property (ii) counts eigenvalues.
Every eigenvalue retained by $P_+$ satisfies $E_i \le S_{vN} + W_\ep$, that is $\lambda_i \ge e^{-(S_{vN} + W_\ep)}$, so if $d_P$ eigenvalues are retained their sum is at least $d_P \, e^{-(S_{vN} + W_\ep)}$.
The retained eigenvalues are a subset of the full spectrum, so their sum is at most one, giving \eqref{rank_bound}.

For (iii), $K = -\log\rho$ commutes with $\rho$, so $P_\pm$, $\rho$ and $\sigma_\pm$ are simultaneously diagonal, $\rho - \sigma_\pm$ is diagonal in the same basis, and its trace norm is the sum of the absolute values of its eigenvalues.
Write $\eta := 1 - \Tr(P_\pm\rho) \le \ep/2$ for the discarded weight.

The states $\sigma_\pm$ carry the in-window eigenvalues rescaled by $(1-\eta)^{-1}$ and zeroes the rest, so the eigenvalues of $\rho - \sigma_\pm$ are $\lambda_i - \lambda_i/(1-\eta)= -\lambda_i\, \eta/(1-\eta)$ inside the window and $\lambda_i$ outside.
Then,
\begin{equation}
    \norm{\rho - \sigma_\pm}_1
    = \frac{\eta}{1-\eta} \sum_{i \in \text{window}} \lambda_i + \sum_{i \notin \text{window}} \lambda_i
    = \frac{\eta}{1-\eta}(1 - \eta) + \eta = 2\eta \le \ep ,
\end{equation}
so $\sigma_\pm \in \Ball{\ep}{\rho}$.

Part (iv) removes that rescaling by capping the spectrum at $t$ instead of renormalizing it.
Let $t := e^{-(S_{vN} - W_\ep)}$ and $(x)_+ := \max(x,0)$, and let $R := \sum_i (\lambda_i - t)_+$ be the weight sitting above the cap.
Since $\lambda_i > t$ is the same condition as $E_i < S_{vN} - W_\ep$, \eqref{ineq_cantelli} gives $R \le \Pr(E < S_{vN} - W_\ep) \le \ep/2$.
Now build $\mu$ in two steps.
First cap, replacing each $\lambda_i$ by $\min(\lambda_i, t)$, which leaves total weight $1 - R$ and every entry at most $t$.
Then return the weight $R$ to entries still below $t$, never letting one exceed $t$.
That second step is possible because the room available is $d\,t - (1 - R)$, which is at least $R$ exactly when $d\,t \ge 1$, and $d\,t \ge e^{W_\ep} \ge 1$ by $S_{vN} \le \log d$.
The result is a state with $\lambda_{\max}(\mu) \le t$, and since it removes weight $R$ and adds the same amount back, $\norm{\rho - \mu}_1 = 2R \le \ep$, which is \eqref{waterfill}.
\end{proof}

Intersecting the two windows gives the two-sided typical subspace of \cite{Bao:2018pvs}, which lies in $\Ball{2\ep}{\rho}$ rather than $\Ball{\ep}{\rho}$ and whose spectrum is flat at leading order.
Its largest and smallest eigenvalues differ by the large factor $e^{2W_\ep}$, but their logarithms agree at leading order in $S_{vN}$, which is the scale entropies are sensitive to.
That state is the analog of the typical subspace of the i.i.d. setting.
Since each argument uses only one tail, the lemma bounds the upper and lower windows separately rather than their intersection.

For a holographic state the half-width may be replaced by the sharper $\widetilde{W}_\ep$ of \eqref{gaussian_window}, which section \ref{sec:dos} derives from the leading-order form \eqref{renyi_leading_order} and which proposition \ref{prop:subgauss} and corollary \ref{cor-window} together prove correct up to the $G_N$-independent additive constant of \eqref{gaussian_window_proved}.
Because $\widetilde{W}_\ep$ is fixed by a two-sided weight of $\ep/2$, it discards $\ep/4$ on each side and therefore satisfies lemma \ref{lem:typical}(i) a fortiori.

Unlike \eqref{window}, which uses nothing beyond the two moments and holds for any state, that form is conditional on \eqref{renyi_leading_order}.
It improves on \eqref{window} exactly when $2\log(4/\ep) < 2/\ep - 1$, whose left side minus right side is decreasing on $\ep \in (0,1)$ and changes sign between $\ep = 1/3$ and $\ep = 2/5$, and it matters when $\ep$ is taken small with $G_N$.
Every statement below that uses it is flagged as conditional, with the assumption-free form quoted alongside.

Applying that truncation in both directions pins the smooth \Renyi entropy at every index from below and from above.
The floor comes from exhibiting a single state in the ball whose largest eigenvalue is small, and the ceiling from showing that no state in the ball can concentrate on a subspace smaller than the modular window, so the two close to within a half-width of $S_{vN}$.

\begin{theorem}
\label{thm:ceiling}
Let $\rho$ be any state on a finite-dimensional Hilbert space, with von Neumann entropy $S_{vN}$ and capacity $C$, and fix $\ep \in (0,1)$.
Then, with $W_\ep$ the half-width \eqref{window}, for every $\alpha > 1$,
\begin{equation}
    S_{vN} - W_\ep
    \;\le\; S^\ep_\alpha(\rho) \;\le\;
    S_{vN} + W_\ep + \frac{\alpha}{\alpha-1}\log\frac{1}{1-\ep} .
    \label{ceiling_bound}
\end{equation}
\end{theorem}

\begin{proof}
For the lower bound, the state $\mu$ of lemma \ref{lem:typical}(iv) lies in $\Ball{\ep}{\rho}$, and by the \Renyi ordering \eqref{renyi_ordering} every $S_\alpha(\mu)$ is at least $S_{\min}(\mu) = -\log\lambda_{\max}(\mu)$, so by \eqref{waterfill}
\begin{equation}
    S^\ep_\alpha(\rho) \;\ge\; S_\alpha(\mu) \;\ge\; -\log\lambda_{\max}(\mu) \;\ge\; S_{vN} - W_\ep .
\end{equation}

For the upper bound, let $\sigma \in \Ball{\ep}{\rho}$ be arbitrary.
Its weight on the upper window is bounded below by \eqref{ineq_proj} together with lemma \ref{lem:typical}(i),
\begin{equation}
    \Tr(P_+\sigma) \;\ge\; \Tr(P_+\rho) - \tfrac{1}{2}\norm{\rho - \sigma}_1 \;\ge\; \lp 1 - \tfrac{\ep}{2} \rp - \tfrac{\ep}{2} \;=\; 1 - \ep .
\end{equation}
Since $P_+$ is a projector of rank $d_P$, $\norm{P_+}_{\alpha/(\alpha-1)} = d_P^{(\alpha-1)/\alpha}$, so \eqref{ineq_holder} at the conjugate pair $\lp \tfrac{\alpha}{\alpha-1}, \alpha \rp$ gives
\begin{equation}
    1 - \ep \;\le\; \Tr(P_+\sigma) \;\le\; \norm{P_+}_{\alpha/(\alpha-1)}\,\norm{\sigma}_\alpha
    \;=\; d_P^{(\alpha-1)/\alpha}\lp\Tr\sigma^\alpha\rp^{1/\alpha} .
\end{equation}
Hence $\Tr\sigma^\alpha \ge (1-\ep)^\alpha d_P^{1-\alpha}$, and applying $\tfrac{1}{1-\alpha}\log$, which reverses the inequality,
\begin{equation}
    S_\alpha(\sigma) \;\le\; \log d_P + \frac{\alpha}{\alpha-1}\log\frac{1}{1-\ep}
    \;\le\; S_{vN} + W_\ep + \frac{\alpha}{\alpha-1}\log\frac{1}{1-\ep}
\end{equation}
by the rank bound \eqref{rank_bound}.
The state $\sigma$ was an arbitrary element of the ball, so the same bound holds for the maximum $S^\ep_\alpha(\rho)$.
\end{proof}

Both steps are established results.
The second display of the proof, read as a statement about an arbitrary test operator rather than about $P_+$, is proposition 4.68 of \cite{KhatriWilde2024Principles} with the identity as reference operator.
They state it for the hypothesis testing relative entropy and prove it by data processing for the sandwiched \Renyi divergence rather than by H\"older. Because that quantity is a supremum over tests, an upper bound on it bounds every admissible test, and $P_+$ is one.
Their proposition contains no window, taking the test as given, and it is the rank bound \eqref{rank_bound} that carries the two moments.
That bound in turn is in substance result 2 of \cite{Boes:2020vpv}.
What the proof above does is put the two together, and the step that is not automatic is that a single projector built from $\rho$ alone certifies every $\sigma$ in the ball simultaneously, which is why \eqref{bnw} cannot simply be chained.

At $\alpha = 2$ the additive term is $-2\log(1-\ep)$ and the H\"older step is Cauchy-Schwarz, the purity being a Hilbert-Schmidt norm. The argument runs the same way at every index.
At $\alpha \to \infty$ the coefficient tends to one and the statement degenerates to the observation that a state with weight $1-\ep$ on a $d_P$-dimensional subspace has $\lambda_{\max} \ge (1-\ep)/d_P$.
Composing the projector construction of \cite{Hayden:unpublished, Bao:2018pvs} with the index-conversion lemma of \cite{Renner2004SmoothRE} gives instead an additive $\tfrac{1}{\alpha-1}\log(2/\ep)$, for $\ep < 2/3$. Their smoothing is over subnormalized truncations of the spectrum rather than the trace-norm ball \eqref{ball_def}, so the comparison below is indicative rather than exact, the conversion between the two notions of ball not being performed here.
The two cross at $\alpha_\star = \log(2/\ep)/\log\tfrac{1}{1-\ep}$, so \eqref{ceiling_bound} is the smaller at every fixed index once $\ep$ is small.
Whichever is quoted, the leading term is $W_\ep$ and the scaling of \eqref{main_result} is unaffected.

Specializing to a holographic state gives the collapse advertised in the introduction.

\begin{corollary}
    \label{cor:flatholo}
    Let $\rho$ be the reduced state of a holographic subregion governed by a single dominant replica saddle on a fixed neighborhood of $\alpha = 1$, so that $C = \mO(S_{vN})$ by \eqref{capacity_def}.
    Then theorem \ref{thm:ceiling} gives $S^\ep_\alpha(\rho) = S_{vN}\lp 1 + \mO(1/\sqrt{S_{vN}})\rp$ at each fixed $\alpha > 1$ and fixed $\ep$, since $W_\ep = \mO(\sqrt{S_{vN}})$ and the additive term is $\mO(1)$.
\end{corollary}

The single-saddle condition does real work.
If two replica saddles exchange dominance at, or within $\mO(G_N)$ of, the identity index, the modular distribution carries two semiclassical peaks separated by $\mO(S_{vN})$ with comparable weights, so that $C \gtrsim S_{vN}^2$, the window is $W_\ep \gtrsim S_{vN}$, and the collapse fails, a peak of weight exceeding $\ep$ being immovable by an $\ep$-perturbation.
Section \ref{sec:discussion} identifies the multi-component regions in which such an exchange happens.

The additive term is uniform for $\alpha$ bounded away from $1$ and diverges as $\alpha \to 1^+$, so the bound is weakest exactly where the smooth \Renyi entropy approaches the von Neumann entropy by continuity in the index.
Thus, the claim \eqref{main_result} is pointwise in $\alpha$ rather than uniform.

Theorem \ref{thm:ceiling} establishes the statement \eqref{main_result} advertised in the introduction, at each fixed $\alpha > 1$.
The same argument applied to $\sigma_\pm$ recovers the flatness of the smooth min- and max-entropies of \cite{Hayden:unpublished, Bao:2018pvs} at the same $\sqrt{S_{vN}}$ scaling.

Theorem \ref{thm:ceiling} does not give a continuity bound. One might hope to recover the bound of section \ref{sec:continuity} from it, since the smoothing ball is itself defined by trace distance, so that a perturbation can be absorbed into a shift of the smoothing radius at no cost in dimension.

Let $\norm{\rho-\sigma}_1 \le \delta$. Since $\Ball{\ep}{\sigma} \subseteq \Ball{\ep+\delta}{\rho}$ and $\Ball{\ep-\delta}{\rho} \subseteq \Ball{\ep}{\sigma}$ by the triangle inequality for the trace norm, maximizing $S_\alpha$ over the larger set in each case gives
\begin{equation}
    S^{\ep - \delta}_\alpha(\rho) \;\le\; S^\ep_\alpha(\sigma) \;\le\; S^{\ep + \delta}_\alpha(\rho)
    \label{absorb}
\end{equation}
for every $\alpha \in (1,\infty]$, where $\alpha = \infty$ is the min-entropy, the upper bound at every $\ep \ge 0$ and the lower at every $\ep \ge \delta$.

Chaining \eqref{absorb} with theorem \ref{thm:ceiling} does replace $\log d$ by $\sqrt{C}$, and it involves no dimension whatsoever, but it cannot produce a continuity bound.
The floor of theorem \ref{thm:ceiling} at $\rho$, then \eqref{absorb} with the two states exchanged, then the ceiling of theorem \ref{thm:ceiling} at $\sigma$ give
\begin{equation}
    S_{vN}(\rho) - W_\ep \;\le\; S^\ep_2(\rho) \;\le\; S^{\ep+\delta}_2(\sigma) \;\le\; S_{vN}(\sigma) + W_{\ep+\delta} - 2\log(1-\ep-\delta) ,
\end{equation}
so that $S_{vN}(\rho) - S_{vN}(\sigma) \le W_\ep + W_{\ep+\delta} - 2\log(1-\ep-\delta)$, and exchanging the two states again bounds the difference in the other direction.
The window $W_\ep$ \textit{grows} as $\ep$ shrinks, so the optimization pushes the smoothing radius to the top of its admissible range, and the two half-widths then stay of order $\sqrt{C}$ however small $\delta$ is.
At $\norm{\rho-\sigma}_1 \le \tfrac14$ and $\ep = \tfrac14$ the chain gives $S_{vN}(\rho) - S_{vN}(\sigma) \le \sqrt{7C_\rho} + \sqrt{3C_\sigma} + \log 4$, and exchanging the two states exchanges the two capacities, so the uniform separation bound is the symmetric
\begin{equation}
    \abs{S_{vN}(\rho) - S_{vN}(\sigma)} \;\le\; \sqrt{7}\lp\sqrt{C_\rho} + \sqrt{C_\sigma}\rp + \log 4 ,
\end{equation}
whose right side stops at order $\sqrt{C}$ instead of vanishing with $\delta$, and which holds only for $\norm{\rho-\sigma}_1 \le \tfrac14$, the radius the chain was run at.
The obstruction is that $S^\ep_2(\rho)$ already sits a distance of order $\sqrt{C}$ from $S_{vN}(\rho)$ at $\delta = 0$, so any chain routed through it inherits that gap.
Theorem \ref{thm:continuity} therefore has to be proved directly on the spectra, as it was in section \ref{sec:continuity}, and is not a corollary of the flatness result.

With the sharpened window of section \ref{sec:dos} the $\ep$ dependence enters only through $\sqrt{\log(1/\ep)}$, so $\ep$ may be taken nonperturbatively small in $G_N$ and the half-width stays subextensive,  provided $\log(1/\ep) = o(S_{vN})$.
The half-width nevertheless grows, since at $\ep \sim e^{-\sqrt{S_{vN}}}$ the window is $\sqrt{2C\sqrt{S_{vN}}} \sim \sqrt{2}\, S_{vN}^{3/4}$, parametrically larger than the $\mO(\sqrt{S_{vN}})$ of \eqref{main_result}.
With the assumption-free window \eqref{window} that choice of $\ep$ gives $W_\ep \sim \sqrt{2C}\,e^{\sqrt{S_{vN}}/2}$, which exceeds $S_{vN}$ and leaves \eqref{ceiling_bound} vacuous, so the nonperturbative $\ep$ relies on holographic structure.


Neither of the constants proved above is optimal.
Optimal smoothing removes the largest eigenvalues carrying weight $\ep/2$, so on a spectrum Gaussian near its peak the deficit approaches the Gaussian quantile at level $\ep$,
\begin{equation}
    S_{vN} - S^\ep_\alpha(\rho) \;\longrightarrow\; \Phi^{-1}(1 - \ep/2) \, \sqrt{C} ,
    \label{sharp_constant}
\end{equation}
at each fixed $\alpha > 1$, since the truncated spectrum has $S_\alpha$ exceeding $-\log\lambda_{\max}$ by only $\mO(\log C)$, so the deficit is the position of the cut up to subleading terms.
The window \eqref{window} overshoots this by $\sqrt{2/\ep - 1}\,/\,\Phi^{-1}(1-\ep/2)$ and the sub-Gaussian half-width by $\sqrt{2\log(4/\ep)}\,/\,\Phi^{-1}(1-\ep/2)$.
The second tends to one as $\ep \to 0$, so the leading-order form \eqref{renyi_leading_order} delivers a half-width that is asymptotically optimal.
The first grows without bound, so no estimate from two moments alone can reach the optimum.
This is a limit rather than a bound, approached slowly from below, and establishing it here would need control of the density of states beyond what section \ref{sec:dos} provides.
It is however established in the i.i.d. setting.
Corollary 4 of \cite{NuradhaWilde2023Fidelity} gives, for every $\alpha > 1$, the second-order expansion of the smooth sandwiched \Renyi relative entropy, which at reference operator $\mbI$ reads
\begin{equation}
    S^\ep_\alpha(\rho^{\otimes n}) \;=\; n S_{vN} + \sqrt{nC}\,\Phi^{-1}(\ep) + \mO(\log n) ,
\end{equation}
independently of the index.
The classical counterpart at $0 < \alpha < 1$ is theorem 1 of \cite{SakaiTan2020Smooth}.
That is \eqref{sharp_constant} with $n$ in place of $1/G_N$, up to the translation between a fidelity ball of subnormalized states and the trace-norm ball \eqref{ball_def}.
Nothing elsewhere in the paper depends on \eqref{sharp_constant}, which is recorded as the value an optimal smoothing would achieve and as the point of contact with the asymptotic literature.

\subsection{Preparation noise cannot perform the smoothing}
\label{sec:noise}

Since every experimental run prepares some state in the ball, it is natural to ask whether state preparation noise performs the smoothing by itself.
Two natural mechanisms do not.
Any unitary supported inside the tested region, or factorizing across the entangling cut, leaves the reduced state isospectral to $\rho$, so coherent errors of that kind cannot smooth at all, smoothing being a spectral operation.
Noise that mixes in a contaminant is controlled by the following bounds.
Neither class exhausts what a preparation can do.
Channels outside them, in particular local noise acting independently on many sites and coherent errors that straddle the entangling surface, are left open.
Engineered filtering, which would smooth, is discussed in section \ref{sec:discussion}.

Write $\sigma = (1-p)\rho + p\,\omega$ for a mixture of that kind, with $\omega$ an arbitrary state and $p \in (0,1)$.
Since $\sigma \ge (1-p)\rho$ as operators, \eqref{ineq_weyl} gives
\begin{equation}
    \lambda^\downarrow_i(\sigma) \;\ge\; (1-p)\lambda^\downarrow_i(\rho) \quad\text{for every } i,
    \qquad\text{whence}\qquad
    \Tr\sigma^\alpha \;\ge\; (1-p)^\alpha \Tr\rho^\alpha
\end{equation}
for $\alpha \ge 1$, and taking $\tfrac{1}{1-\alpha}\log$ of both sides reverses the inequality to give, at every $\alpha > 1$,
\begin{equation}
    S_\alpha(\sigma) \;\le\; S_\alpha(\rho) + \frac{\alpha}{\alpha-1}\log\frac{1}{1-p} .
    \label{noise_upper}
\end{equation}
In the other direction the convexity \eqref{ineq_convex} gives
\begin{equation}
    \Tr\sigma^\alpha \;\le\; (1-p)\Tr\rho^\alpha + p\Tr\omega^\alpha .
\end{equation}
If $S_\alpha(\omega) \ge S_\alpha(\rho) - \Delta_\alpha$ for some $\Delta_\alpha \ge 0$, which reads $\Tr\omega^\alpha \le e^{(\alpha-1)\Delta_\alpha}\Tr\rho^\alpha$, then
\begin{equation}
    S_\alpha(\sigma) \;\ge\; S_\alpha(\rho) - \frac{1}{\alpha-1}\log\lp 1 + p \lp e^{(\alpha-1)\Delta_\alpha} - 1 \rp \rp .
    \label{noise_lower}
\end{equation}
Reaching the optimizer of \eqref{smooth_entropy_def} requires raising $S_2$ by an amount extensive in $c$, whereas \eqref{noise_upper} caps the rise at $\tfrac{\alpha}{\alpha-1}\log\tfrac{1}{1-p}$.
That cap is $\mO(1)$ only for $p$ bounded away from one, which ball membership alone does not deliver, since $\rho - \sigma = p(\rho-\omega)$ constrains the product $p\norm{\rho-\omega}_1$ rather than $p$ itself.
A separation $\norm{\rho-\omega}_1 \ge \kappa$ between target and contaminant restores it, forcing $p \le \ep/\kappa$ and capping the rise at $\mO(\ep/\kappa)$, and global depolarizing noise supplies such a $\kappa$ by replacing $\rho$ on the tested region with something macroscopically different from it. A Pauli error supported inside the region supplies none, since it acts there as a unitary and leaves the reduced state isospectral, which is the case already disposed of above.
Under that assumption the measured moments of a noisy preparation are those of a state whose purity differs from that of $\rho$ by $\mO(1)$ factors, while the smoothing optimum occupies a corner of the ball in which the purity is suppressed by $e^{-\mO(c)}$, so the optimization \eqref{smooth_entropy_def} has to be solved once analytically, as it is in theorem \ref{thm:ceiling}, rather than run on the data.
The lower bound \eqref{noise_lower} controls the opposite direction, where the damage a contaminant does is governed by $\Delta_\alpha$, by how much \textit{less entangled} it is on the tested region than the target, and not by how close to pure it is.
A flat contaminant of rank $e^{\beta S_{vN}}$ is exponentially far from pure and still dominates the measured moment at index $\alpha$ whenever $\beta$ falls below $S_\alpha(\rho)/S_{vN}$, which for a vacuum interval at $\beta = 5/8$ means every $\alpha < 4$.

\section{Discussion}
\label{sec:discussion}

The capacity of entanglement controls two distinct things.
The first is the von Neumann entropy itself.
The Ryu-Takayanagi formula assigns an entropy to a geometry, and how sharply it does so is the question theorem \ref{thm:continuity} answers.
Fannes-Audenaert answers it with the dimension of the regulated Hilbert space, which is a property of the cutoff and not of the bulk, and which diverges in the continuum however close the two states are.
Theorem \ref{thm:continuity} answers it with the capacity, whose square root is the spread of the modular energy.
For a holographic state the fixed-area decomposition $\rho \simeq \bigoplus_A p_A \rho_A$ with each $\rho_A$ flat gives $E = A/4G_N - \log p_A$, so at leading order in $G_N$ that spread is the spread of the area itself, the contribution of $-\log p_A$ being subleading for a smooth area distribution \cite{DeBoer:2018kvc, Dong:2018seb}.
What the theorem says is that the precision with which a geometry fixes an entropy is set by how much its minimal surface fluctuates.

Both capacities enter \eqref{capacity_continuity}, so the regimes below are statements about pairs of states rather than about one state, and the ball as a whole is not among them. Section \ref{sec:continuity} settles what the ball costs, in terms of a threshold on the capacity of the contaminant.
Three regimes follow.
For two states of vanishing capacity, fixed-area states among them, \eqref{capacity_continuity} collapses to $-\log(1-\delta/2)$, which is of order $\delta$ and carries no dimension at all, so the entropy is Lipschitz across that family even in the continuum. The nested flat states of section \ref{sec:continuity} saturate this exactly.
For two states of capacity $\mO(S_{vN})$, holographic states among them, the precision is $\mO(\sqrt{\delta S_{vN}})$, which is subextensive, so a leading-order statement about the geometry survives a perturbation within that class.
Where two minimal surfaces compete the area spread is itself $\mO(S_{vN})$, the capacity is $C \gtrsim S_{vN}^2$, and the precision degrades to $\mO(\sqrt{\delta}\,S_{vN})$, which is extensive.
The dictionary is sharpest where the surface is rigid and fails where it is degenerate, and it degrades in proportion to the area spread of the superposition in between. That is a reason to prefer the fixed-area basis for holographic entropy statements which is independent of the observation below that fixed-area states need no smoothing.

The second is the smooth \Renyi spectrum, and there the area has to mean something for a single state.
The operational content of the von Neumann entropy, as a compression rate or a dilution rate, is defined on asymptotically many copies, and a semiclassical geometry describes one state, with no tensor power of it available.
Theorem \ref{thm:ceiling} supplies the reason the distinction does not matter at large central charge.
Under the single-saddle condition of corollary \ref{cor:flatholo}, every smooth \Renyi entropy at fixed $\alpha > 1$ lies within $\mO(\sqrt{S_{vN}})$ of $S_{vN}$, and the smooth min- and max-entropies with it.
The area is therefore not merely the von Neumann entropy of the boundary state.
To leading order in $1/G_N$ it is every one-shot entropy of that state at once, and the geometry does not distinguish them.

The mechanism is the one recorded in section \ref{sec:dos}.
The low cumulants of the modular energy scale as $1/G_N$, which is the structure a sum of $1/G_N$ independent variables has, so the semiclassical limit shares the concentration rate of an asymptotic limit with the inverse Newton constant in place of the copy number. One feature of the i.i.d. case is not shared, and it is the reason theorem \ref{thm:continuity} is needed at all. For $n$ copies $\log d^{\otimes n}$ also grows like $n$, so a dimensional continuity bound is already adequate per copy. In a continuum theory $\log d / S_{vN}$ diverges as the regulator is removed, and no dimensional bound survives it.
That is why holography can describe a single state with an asymptotic information measure. Equation \eqref{aep} states the comparison with the asymptotic equipartition property directly.
The same exchange of saddles that breaks the holographic reading of theorem \ref{thm:continuity} breaks this one, since a bimodal modular distribution has an extensive window and pins no one-shot entropy to the area.
Both halves of the paper therefore fail in the same place, at an entanglement transition, and for the same reason, as the failure mode recorded below sets out.
The bound is two-sided at every index and consumes only the mean and variance of the modular energy.
For holographic states the half-width tightens to the sub-Gaussian form of proposition \ref{prop:subgauss}, which holds whenever $\log(1/\ep)$ is small compared with $S_{vN}$, and that is the range in which the flatness statement carries content.
The collapse identifies a single subextensive band containing both $S_{vN}$ and every smooth \Renyi entropy at fixed index, and it caps from above what preparation imperfection can do to any measured moment.
It does not place the bare measured quantities in that band, and cannot, since the bare \Renyitwo entropy of a vacuum interval sits a distance $S_{vN}/4$ away by \eqref{cardy_renyi}.
In that operational sense theorem \ref{thm:ceiling} is one-sided, though the inequality itself is two-sided.
Its ceiling holds for every state in the ball and therefore does bound what an apparatus can be shown, while its floor is exhibited by the single constructed state $\mu$ of lemma \ref{lem:typical}(iv) and says nothing about the state an apparatus holds.
The same metric ball carries two distinct readings that should not be interchanged, one-shot task smoothing, in which a nearby state is selected as part of a protocol with a tolerated failure probability, and preparation uncertainty, in which the physical state is unknown within a confidence region. The smooth entropy is an optimized functional over the ball rather than anything measured, and it is the second reading that the ceiling serves.

\paragraph{Fixed-area states.}
In fixed-area eigenstates of the gravitational constraints the bare \Renyi spectrum is already flat at leading order \cite{Akers:2018fow, Dong:2018seb}, so for such states no smoothing step is needed and the bare \Renyitwo entropy already stands in for $S_{vN}$.
Generic holographic states are superpositions over areas, and the analysis of section \ref{sec:flatness} is what undoes the spread induced by that superposition, in line with the interpretation of tensor networks as fixed-area backgrounds \cite{Bao:2018pvs, Akers:2018fow, Dong:2018seb}.

\paragraph{The entropy cone.}
The cone is the sharpest state-independent signature of a semiclassical bulk that is checkable on finitely many numbers. Its inequalities hold for every geometric state and fail for entanglement patterns, GHZ-type correlations among three or more parties, that no classical bulk can produce, so a measured entropy vector lying outside it would rule out a geometric dual from boundary data alone.
The application that motivates the flatness result is a finite-resource test for a semiclassical bulk dual, run through the cone \cite{Bao:2015bfa}. It would assemble a proxy vector $\hat{S} = (\hat{S}(I))_I$ over the subsets $I$ of a partition and ask whether it lies in the cone, a vector outside the cone being the informative outcome.
The conditional statement the present results support is that \textit{if} the state is holographic, \textit{then} by corollary \ref{cor:flatholo} every smooth \Renyi entropy of every subset agrees with $S_{vN}(\rho_I)$ up to $\mO(\sqrt{S_{vN}})$, and the proxy vector lies inside the cone up to subextensive displacements.

The contrapositive, that a proxy vector outside the cone by an extensive margin cannot have come from a holographic state, is the discriminator, and theorem \ref{thm:ceiling} makes it quantitative.

\begin{corollary}
    \label{cor:cone}
    Let $\sum_I q_I S(I) \ge 0$ be an inequality valid for every holographic entropy vector, and let $\rho$ be a state whose von Neumann entropy vector satisfies it, $\sum_I q_I S_{vN}(\rho_I) \ge 0$, and each of whose marginals satisfies the conditions of corollary \ref{cor:flatholo}. Then for every $\alpha > 1$ and every choice of radii $\ep_I \in (0,1)$,
    \begin{equation}
        \sum_I q_I S^{\ep_I}_\alpha(\rho_I) \;\ge\; -\sum_I \abs{q_I} \lb W_{\ep_I}(\rho_I) + \frac{\alpha}{\alpha-1}\log\frac{1}{1-\ep_I} \rb .
        \label{cone_robust}
    \end{equation}
\end{corollary}

\noindent This follows from
\begin{equation}
    \abs{S^{\ep_I}_\alpha(\rho_I) - S_{vN}(\rho_I)} \;\le\; W_{\ep_I} + \frac{\alpha}{\alpha-1}\log\frac{1}{1-\ep_I}
\end{equation}
applied termwise to the assumed $\sum_I q_I S_{vN}(\rho_I) \ge 0$. The hypothesis is stated on the von Neumann entropies rather than inferred from geometricity, since a holographic entropy vector is a vector of areas over $4G_N$ and the exact entropies depart from it at $\mO(1)$ through the corrections discussed below.
Its right side is $\mO(\sqrt{S_{vN}})$ at fixed party number, the sum over subsets contributing a constant that grows combinatorially in the number of parties, so a smooth \Renyi vector violating a cone inequality by an extensive margin cannot come from a holographic state.
Note however that the corollary is not a test, since the proxy vector $\hat{S}$ is never constructed here. It bounds the apparent violation by controlling each entropy in the vector separately and adding the errors, which leaves estimation and certification as the two things still missing.
Qualifications of the cone itself would have to be absorbed into any decision rule.
It constrains entropies only at leading order in $1/G_N$, with quantum extremal surface corrections entering at $\mO(1)$ and capable of pushing entropy vectors outside it \cite{Faulkner:2013ana, Engelhardt:2014gca, Akers:2020pmf, Akers:2021lms}, so only violations extensive in $c$ are meaningful.
And beyond monogamy of mutual information, which holds covariantly \cite{Hayden:2011ag, Wall:2012uf}, the higher-party inequalities are established for static states, though evidence has accumulated that the RT and HRT cones coincide \cite{Czech:2019lps, Caginalp:2019mgu, Grado-White:2024gtx, Grado-White:2025jci, Grimaldi:2025jad, Grimaldi:2026lbq, Czech:2026zca}.

\paragraph{Estimation at finite copies.}
Converting the flatness result into an estimator of $S_{vN}$ from a handful of measured integer moments is the natural next step, and the natural route is a controlled continuation in the \Renyi index from the measured integers down to $\alpha = 1$.
Conditioning and certification both stand between that route and a usable protocol, and neither is settled here.
The first is conditioning.
Any extrapolation from indices $\alpha \ge 2$ to $\alpha = 1$ amplifies the error on its inputs, and the strategy it has to beat is the null one of quoting the bare \Renyitwo entropy and accepting the deficit \eqref{cardy_renyi}, so what matters is whether the amplified preparation and statistical errors can be held below that deficit at entropies a measurement can reach.
The exponential cost of the higher moments enters on the same side of that comparison.
The second is certification.
The flat contaminant described after \eqref{noise_lower} dominates the low-index moments while remaining a physically unremarkable preparation error, and by theorem \ref{thm:continuity} any state in the ball whose entropy differs from the target at leading order carries a correspondingly large capacity.

Theorem \ref{thm:continuity} makes part of this actionable. Its right side depends on the unknown state only through $C_\sigma$, and the capacity is a second cumulant of the same modular distribution the replica moments already sample, so it is estimable from the same data as the entropy. An anomalously large measured capacity is therefore a signal that the entropy estimate cannot be trusted, and a small one vindicates it, which makes the pair $(\hat{S}, \hat{C})$ self-certifying in a way $\hat{S}$ alone is not. What remains open is quantitative, whether the capacity can be estimated well enough, from the same moments and at the same cost, to make the criterion bite at achievable entropies.

\paragraph{A failure mode of the smoothness assumption.}
The assumption behind \eqref{renyi_leading_order}, that $s_\alpha$ is smooth in $\alpha$, carries a risk for multi-component regions, since for unions of regions the dominant replica saddle, equivalently the dominant cosmic-brane configuration, can exchange as a function of $\alpha$.
For two intervals this is visible already in the mutual-information transition, whose critical cross-ratio depends on the \Renyi index \cite{Headrick:2010zt, Faulkner:2013yia}.
Near such a transition the branch of $s_\alpha$ that dominates at integer $\alpha \ge 2$ need not be the branch controlling $\alpha \to 1$, and both branches are smooth at integer $\alpha$, so no predictive check detects the exchange.
Theorem \ref{thm:ceiling} is untouched by this, since it holds for an arbitrary state and assumes nothing about $s_\alpha$.
What an exchange at, or within $\mO(G_N)$ of, the identity index threatens is the holographic input $C = \mO(S_{vN})$ itself, and with it both corollary \ref{cor:flatholo} and the holographic reading of theorem \ref{thm:continuity}. Theorem \ref{thm:continuity} is a statement about arbitrary states and is untouched, but the conclusion drawn from it, that the ambiguity of $S_{vN}$ over a preparation ball is $\mO(\sqrt{\delta S_{vN}})$ and therefore subextensive, uses the same input and fails with it.
At such a point the modular spectrum spreads over a range of order $S_{vN}$, two peaks of comparable weight being the extreme case, so $C \gtrsim S_{vN}^2$ and the window $W_\ep$ becomes extensive. The collapse then genuinely fails rather than merely going unproved, since a low-energy peak of weight exceeding $\ep$ cannot be removed by an $\ep$-perturbation and pins the smooth \Renyi entropy near that peak.
Read forwards rather than as a caveat, this is a prediction. The capacity of entanglement diverges relative to the entropy at an entanglement transition, which is a sharp statement about holographic states that tensor-network models and two-interval computations can test independently of anything else in this paper.
The failure window in the action gap between the two saddles is $\mO(G_N)$ wide, so the configuration is fine-tuned. It is not exotic, since degenerate surfaces are what saturating a cone inequality tends to involve, so the discriminator of corollary \ref{cor:cone} is least trustworthy exactly where it would be used. That is a structural obstruction to any finite-resource cone test and it is independent of estimation and certification.
What the exchange threatens is the saddle-point density of states of section \ref{sec:dos}, whose Legendre transform needs smoothness at every index, and any continuation from integer moments to $\alpha = 1$, which would follow the wrong branch with an $\mO(c)$ error.
Proposition \ref{prop:subgauss} is insensitive to an exchange away from the identity index, since it evaluates $s_\alpha$ only on $[1-t_0, 1+t_0]$, and it fails only if the exchange occurs inside that neighborhood.
The danger is confined to configurations tuned near an entanglement transition and is mitigated by varying the geometry, since an exchange point moves with the cross-ratios while a genuine physical effect should be robust.

\paragraph{Engineered smoothing.}
The bound \eqref{noise_upper} shows that preparation noise of mixture type cannot realize the smoothing optimum whenever the contaminant is bounded away from the target, $\norm{\rho-\omega}_1 \ge \kappa$, which is why the smoothing has to be solved analytically.
Without such a separation the statement is empty, since any state in the ball is the mixture at $p = 1$ and $\omega = \sigma$.
A hardware realization would require engineered modular-energy filtering, for example weak measurement of the subsystem modular energy followed by post-selection.
Extending \eqref{noise_upper} beyond mixture channels, in particular to noise acting independently at each site, is a natural next step.

\paragraph{Beyond leading order.}
The window width $\mO(\sqrt{S_{vN}})$ is a floor on the accuracy of any estimator built on the flatness result, and it is far above the $\mO(1)$ scale at which quantum corrections to the RT formula live.
Reaching that scale would require estimators beating the smoothing floor, presumably built on the one-shot quantities that govern holographic entropies at that order \cite{Akers:2020pmf}, and appears difficult.

\paragraph{Outlook.}
The regime where these bounds have content lies far above the platforms where \Renyi moments are measured routinely \cite{Islam:2015mom, Linke:2017xlv, Brydges:2019wut, Elben:2022jvo, Landsman:2018jpm} and above the numerical methods where replica moments are standard \cite{Hastings:2010zka, Humeniuk:2012xg}.
The results of this paper are therefore asymptotic statements about what is well posed, and closing that gap is a prerequisite for any near-term application.
The smooth \Renyitwo entropy is the quantity a measurement can bound from above, and by corollary \ref{cor:flatholo} it agrees with the RT area to $\mO(\sqrt{S_{vN}})$. One-shot information theory is therefore the right language for holographic entropy at large but finite $c$.

\acknowledgments

Ning Bao acknowledges support from Northeastern University and from the ASCR EXPRESS project
Quantum Transforms from Classical Transforms.
J.M. is supported by a graduate research assistantship from Northeastern University.
Claude models were used in the conduct of this research, for theorem development. The authors
are responsible for all statements in this paper.

\bibliographystyle{JHEP}
\bibliography{main}

@article{Abanin:2012jms,
    author = "Abanin, Dmitry A. and Demler, Eugene",
    title = "{Measuring Entanglement Entropy of a Generic Many-Body System with a Quantum Switch}",
    eprint = "1204.2819",
    archivePrefix = "arXiv",
    primaryClass = "cond-mat.mes-hall",
    doi = "10.1103/PhysRevLett.109.020504",
    journal = "Phys. Rev. Lett.",
    volume = "109",
    number = "2",
    pages = "020504",
    year = "2012"
}

@article{Akers:2018fow,
    author = "Akers, Chris and Rath, Pratik",
    title = "{Holographic Renyi Entropy from Quantum Error Correction}",
    eprint = "1811.05171",
    archivePrefix = "arXiv",
    primaryClass = "hep-th",
    doi = "10.1007/JHEP05(2019)052",
    journal = "JHEP",
    volume = "05",
    pages = "052",
    year = "2019"
}

@article{Akers:2020pmf,
    author = "Akers, Chris and Penington, Geoff",
    title = "{Leading order corrections to the quantum extremal surface prescription}",
    eprint = "2008.03319",
    archivePrefix = "arXiv",
    primaryClass = "hep-th",
    doi = "10.1007/JHEP04(2021)062",
    journal = "JHEP",
    volume = "04",
    pages = "062",
    year = "2021"
}

@article{Akers:2021lms,
    author = "Akers, Chris and Hern{\'a}ndez-Cuenca, Sergio and Rath, Pratik",
    title = "{Quantum Extremal Surfaces and the Holographic Entropy Cone}",
    eprint = "2108.07280",
    archivePrefix = "arXiv",
    primaryClass = "hep-th",
    doi = "10.1007/JHEP11(2021)177",
    journal = "JHEP",
    volume = "11",
    pages = "177",
    year = "2021"
}

@article{Akers:2023fqr,
    author = "Akers, Chris and Levine, Adam and Penington, Geoff and Wildenhain, Elizabeth",
    title = "{One-shot holography}",
    eprint = "2307.13032",
    archivePrefix = "arXiv",
    primaryClass = "hep-th",
    doi = "10.21468/SciPostPhys.16.6.144",
    journal = "SciPost Phys.",
    volume = "16",
    number = "6",
    pages = "144",
    year = "2024"
}

@article{Audenaert:2006vjl,
    author = "Audenaert, Koenraad M. R.",
    title = "{A sharp continuity estimate for the von Neumann entropy}",
    eprint = "quant-ph/0610146",
    archivePrefix = "arXiv",
    doi = "10.1088/1751-8113/40/28/S18",
    journal = "J. Phys. A",
    volume = "40",
    number = "28",
    pages = "8127",
    year = "2007"
}

@article{Audenaert2025ContinuityBounds,
    author = "Audenaert, Koenraad and Bergh, Bjarne and Datta, Nilanjana and Jabbour, Michael G. and Capel, \'{A}ngela and Gondolf, Paul",
    title = "{Continuity bounds for quantum entropies arising from a fundamental entropic inequality}",
    eprint = "2408.15306",
    archivePrefix = "arXiv",
    primaryClass = "quant-ph",
    doi = "10.1109/TIT.2025.3586478",
    journal = "IEEE Trans. Inf. Theor.",
    volume = "71",
    number = "9",
    pages = "7029--7038",
    year = "2025"
}

@article{Bao:2015bfa,
    author = "Bao, Ning and Nezami, Sepehr and Ooguri, Hirosi and Stoica, Bogdan and Sully, James and Walter, Michael",
    title = "{The Holographic Entropy Cone}",
    eprint = "1505.07839",
    archivePrefix = "arXiv",
    primaryClass = "hep-th",
    reportNumber = "CALT-TH-2015-020, IPMU15-0074, SLAC-PUB-16294, SU-ITP-15-08, CALT-TH 2015-020, IPMU15-0074, SLAC-PUB-16294, SU-ITP-15/08",
    doi = "10.1007/JHEP09(2015)130",
    journal = "JHEP",
    volume = "09",
    pages = "130",
    year = "2015"
}

@article{Bao:2018pvs,
    author = "Bao, Ning and Penington, Geoffrey and Sorce, Jonathan and Wall, Aron C.",
    title = "{Beyond Toy Models: Distilling Tensor Networks in Full AdS/CFT}",
    eprint = "1812.01171",
    archivePrefix = "arXiv",
    primaryClass = "hep-th",
    doi = "10.1007/JHEP11(2019)069",
    journal = "JHEP",
    volume = "11",
    pages = "069",
    year = "2019"
}

@article{Bao:2025plr,
    author = "Bao, Ning and Geng, Hao and Jiang, Yikun",
    title = "{Ryu-Takayanagi formula for multi-boundary black holes from 2D large-c CFT ensemble}",
    eprint = "2504.12388",
    archivePrefix = "arXiv",
    primaryClass = "hep-th",
    doi = "10.1007/JHEP10(2025)042",
    journal = "JHEP",
    volume = "10",
    pages = "042",
    year = "2025"
}

@article{Berta2025IntegralRepresentations,
    author = "Berta, Mario and Lami, Ludovico and Tomamichel, Marco",
    title = "{Continuity of entropies via integral representations}",
    eprint = "2408.15226",
    archivePrefix = "arXiv",
    primaryClass = "quant-ph",
    doi = "10.1109/TIT.2025.3527858",
    journal = "IEEE Trans. Inf. Theor.",
    volume = "71",
    number = "3",
    pages = "1896--1908",
    year = "2025"
}

@article{Boes:2020vpv,
    author = "Boes, Paul and Ng, Nelly H. Y. and Wilming, Henrik",
    title = "{The variance of relative surprisal as single-shot quantifier}",
    eprint = "2009.08391",
    archivePrefix = "arXiv",
    primaryClass = "quant-ph",
    doi = "10.1103/PRXQuantum.3.010325",
    journal = "PRX Quantum",
    volume = "3",
    number = "1",
    pages = "010325",
    year = "2022"
}

@article{Brydges:2019wut,
    author = "Brydges, Tiff and Elben, Andreas and Jurcevic, Petar and Vermersch, Beno{\^\i}t and Maier, Christine and Lanyon, Ben P. and Zoller, Peter and Blatt, Rainer and Roos, Christian F.",
    title = "{Probing R{\'e}nyi entanglement entropy via randomized measurements}",
    eprint = "1806.05747",
    archivePrefix = "arXiv",
    primaryClass = "quant-ph",
    doi = "10.1126/science.aau4963",
    journal = "Science",
    volume = "364",
    number = "6437",
    pages = "aau4963",
    year = "2019"
}

@article{Caginalp:2019mgu,
    author = "Caginalp, Reginald J.",
    title = "{Holographic entropy cone in AdS-Vaidya spacetimes}",
    eprint = "1905.00544",
    archivePrefix = "arXiv",
    primaryClass = "hep-th",
    doi = "10.1103/PhysRevD.101.026010",
    journal = "Phys. Rev. D",
    volume = "101",
    number = "2",
    pages = "026010",
    year = "2020"
}

@article{Calabrese:2004eu,
    author = "Calabrese, Pasquale and Cardy, John L.",
    title = "{Entanglement entropy and quantum field theory}",
    eprint = "hep-th/0405152",
    archivePrefix = "arXiv",
    doi = "10.1088/1742-5468/2004/06/P06002",
    journal = "J. Stat. Mech.",
    volume = "0406",
    pages = "P06002",
    year = "2004"
}

@article{Cardy:2011zz,
    author = "Cardy, John",
    title = "{Measuring Entanglement Using Quantum Quenches}",
    eprint = "1012.5116",
    archivePrefix = "arXiv",
    primaryClass = "cond-mat.stat-mech",
    reportNumber = "NSF-KITP-10-164",
    doi = "10.1103/PhysRevLett.106.150404",
    journal = "Phys. Rev. Lett.",
    volume = "106",
    pages = "150404",
    year = "2011"
}

@article{Czech:2014tva,
    author = "Czech, Bartlomiej and Hayden, Patrick and Lashkari, Nima and Swingle, Brian",
    title = "{The Information Theoretic Interpretation of the Length of a Curve}",
    eprint = "1410.1540",
    archivePrefix = "arXiv",
    primaryClass = "hep-th",
    doi = "10.1007/JHEP06(2015)157",
    journal = "JHEP",
    volume = "06",
    pages = "157",
    year = "2015"
}

@article{Czech:2019lps,
    author = "Czech, Bartlomiej and Dong, Xi",
    title = "{Holographic Entropy Cone with Time Dependence in Two Dimensions}",
    eprint = "1905.03787",
    archivePrefix = "arXiv",
    primaryClass = "hep-th",
    doi = "10.1007/JHEP10(2019)177",
    journal = "JHEP",
    volume = "10",
    pages = "177",
    year = "2019"
}

@article{Czech:2026zca,
    author = "Czech, Bartlomiej and Feng, Yichen and Wu, Xianlai and Xie, Minjun",
    title = "{Holographic entropy inequalities pass the majorization test}",
    eprint = "2601.09989",
    archivePrefix = "arXiv",
    primaryClass = "hep-th",
    month = "1",
    year = "2026"
}

@article{Daley:2012xhf,
    author = "Daley, A. J. and Pichler, H. and Schachenmayer, J. and Zoller, P.",
    title = "{Measuring Entanglement Growth in Quench Dynamics of Bosons in an Optical Lattice}",
    eprint = "1205.1521",
    archivePrefix = "arXiv",
    primaryClass = "cond-mat.quant-gas",
    doi = "10.1103/PhysRevLett.109.020505",
    journal = "Phys. Rev. Lett.",
    volume = "109",
    number = "2",
    pages = "020505",
    year = "2012"
}

@article{DeBoer:2018kvc,
    author = {De Boer, Jan and J{\"a}rvel{\"a}, Jarkko and Keski-Vakkuri, Esko},
    title = "{Aspects of capacity of entanglement}",
    eprint = "1807.07357",
    archivePrefix = "arXiv",
    primaryClass = "hep-th",
    reportNumber = "HIP-2018-25/TH",
    doi = "10.1103/PhysRevD.99.066012",
    journal = "Phys. Rev. D",
    volume = "99",
    number = "6",
    pages = "066012",
    year = "2019"
}

@article{Dong:2016fnf,
    author = "Dong, Xi",
    title = "{The Gravity Dual of Renyi Entropy}",
    eprint = "1601.06788",
    archivePrefix = "arXiv",
    primaryClass = "hep-th",
    reportNumber = "SU-ITP-16/01, SU-ITP-16-01",
    doi = "10.1038/ncomms12472",
    journal = "Nature Commun.",
    volume = "7",
    pages = "12472",
    year = "2016"
}

@article{Dong:2016hjy,
    author = "Dong, Xi and Lewkowycz, Aitor and Rangamani, Mukund",
    title = "{Deriving covariant holographic entanglement}",
    eprint = "1607.07506",
    archivePrefix = "arXiv",
    primaryClass = "hep-th",
    doi = "10.1007/JHEP11(2016)028",
    journal = "JHEP",
    volume = "11",
    pages = "028",
    year = "2016"
}

@article{Dong:2018seb,
    author = "Dong, Xi and Harlow, Daniel and Marolf, Donald",
    title = "{Flat entanglement spectra in fixed-area states of quantum gravity}",
    eprint = "1811.05382",
    archivePrefix = "arXiv",
    primaryClass = "hep-th",
    doi = "10.1007/JHEP10(2019)240",
    journal = "JHEP",
    volume = "10",
    pages = "240",
    year = "2019"
}

@article{Ekert:2002qtj,
    author = "Ekert, Artur K. and Alves, Carolina Moura and Oi, Daniel K. L. and Horodecki, Micha{\l} and Horodecki, Pawe{\l} and Kwek, L. C.",
    title = "{Direct Estimations of Linear and Nonlinear Functionals of a Quantum State}",
    eprint = "quant-ph/0203016",
    archivePrefix = "arXiv",
    doi = "10.1103/PhysRevLett.88.217901",
    journal = "Phys. Rev. Lett.",
    volume = "88",
    number = "21",
    pages = "217901",
    year = "2002"
}

@article{Elben:2022jvo,
    author = "Elben, Andreas and Flammia, Steven T. and Huang, Hsin-Yuan and Kueng, Richard and Preskill, John and Vermersch, Beno{\^\i}t and Zoller, Peter",
    title = "{The randomized measurement toolbox}",
    eprint = "2203.11374",
    archivePrefix = "arXiv",
    primaryClass = "quant-ph",
    doi = "10.1038/s42254-022-00535-2",
    journal = "Nature Rev. Phys.",
    volume = "5",
    number = "1",
    pages = "9--24",
    year = "2023"
}

@article{Engelhardt:2014gca,
    author = "Engelhardt, Netta and Wall, Aron C.",
    title = "{Quantum Extremal Surfaces: Holographic Entanglement Entropy beyond the Classical Regime}",
    eprint = "1408.3203",
    archivePrefix = "arXiv",
    primaryClass = "hep-th",
    doi = "10.1007/JHEP01(2015)073",
    journal = "JHEP",
    volume = "01",
    pages = "073",
    year = "2015"
}

@article{Faulkner:2013ana,
    author = "Faulkner, Thomas and Lewkowycz, Aitor and Maldacena, Juan",
    title = "{Quantum corrections to holographic entanglement entropy}",
    eprint = "1307.2892",
    archivePrefix = "arXiv",
    primaryClass = "hep-th",
    doi = "10.1007/JHEP11(2013)074",
    journal = "JHEP",
    volume = "11",
    pages = "074",
    year = "2013"
}

@article{Faulkner:2013yia,
    author = "Faulkner, Thomas",
    title = "{The Entanglement Renyi Entropies of Disjoint Intervals in AdS/CFT}",
    eprint = "1303.7221",
    archivePrefix = "arXiv",
    primaryClass = "hep-th",
    month = "3",
    year = "2013"
}

@article{Grado-White:2024gtx,
    author = "Grado-White, Brianna and Grimaldi, Guglielmo and Headrick, Matthew and Hubeny, Veronika E.",
    title = "{Testing holographic entropy inequalities in 2 + 1 dimensions}",
    eprint = "2407.07165",
    archivePrefix = "arXiv",
    primaryClass = "hep-th",
    reportNumber = "BRX-TH-6721",
    doi = "10.1007/JHEP01(2025)065",
    journal = "JHEP",
    volume = "01",
    pages = "065",
    year = "2025"
}

@article{Grado-White:2025jci,
    author = "Grado-White, Brianna and Grimaldi, Guglielmo and Headrick, Matthew and Hubeny, Veronika E.",
    title = "{Minimax surfaces and the holographic entropy cone}",
    eprint = "2502.09894",
    archivePrefix = "arXiv",
    primaryClass = "hep-th",
    doi = "10.1007/JHEP05(2025)104",
    journal = "JHEP",
    volume = "05",
    pages = "104",
    year = "2025"
}

@article{Grimaldi:2025jad,
    author = "Grimaldi, Guglielmo and Headrick, Matthew and Hubeny, Veronika E.",
    title = "{A new characterization of the holographic entropy cone}",
    eprint = "2508.21823",
    archivePrefix = "arXiv",
    primaryClass = "hep-th",
    doi = "10.21468/SciPostPhys.20.4.122",
    journal = "SciPost Phys.",
    volume = "20",
    number = "4",
    pages = "122",
    year = "2026"
}

@article{Grimaldi:2026lbq,
    author = "Grimaldi, Guglielmo and Headrick, Matthew and Hubeny, Veronika E. and Shteyner, Pavel",
    title = "{Combinatorial properties of holographic entropy inequalities}",
    eprint = "2601.09987",
    archivePrefix = "arXiv",
    primaryClass = "hep-th",
    month = "1",
    year = "2026"
}

@article{Hastings:2010zka,
    author = "Hastings, Matthew B. and Gonz{\'a}lez, Iv{\'a}n and Kallin, Ann B. and Melko, Roger G.",
    title = "{Measuring Renyi Entanglement Entropy in Quantum Monte Carlo Simulations}",
    eprint = "1001.2335",
    archivePrefix = "arXiv",
    primaryClass = "cond-mat.str-el",
    doi = "10.1103/PhysRevLett.104.157201",
    journal = "Phys. Rev. Lett.",
    volume = "104",
    number = "15",
    pages = "157201",
    year = "2010"
}

@article{Hayden:2011ag,
    author = "Hayden, Patrick and Headrick, Matthew and Maloney, Alexander",
    title = "{Holographic Mutual Information is Monogamous}",
    eprint = "1107.2940",
    archivePrefix = "arXiv",
    primaryClass = "hep-th",
    reportNumber = "BRX-TH-638, BRX-TH-638",
    doi = "10.1103/PhysRevD.87.046003",
    journal = "Phys. Rev. D",
    volume = "87",
    number = "4",
    pages = "046003",
    year = "2013"
}

@unpublished{Hayden:unpublished,
    author = "Hayden, Patrick and Swingle, Brian and Walter, Michael",
    title = "{One-shot information theory in quantum field theory}",
    note = "unpublished"
}

@article{Headrick:2010zt,
    author = "Headrick, Matthew",
    title = "{Entanglement Renyi entropies in holographic theories}",
    eprint = "1006.0047",
    archivePrefix = "arXiv",
    primaryClass = "hep-th",
    reportNumber = "BRX-TH-619, BRX-TH 619",
    doi = "10.1103/PhysRevD.82.126010",
    journal = "Phys. Rev. D",
    volume = "82",
    pages = "126010",
    year = "2010"
}

@article{Holzhey:1994we,
    author = "Holzhey, Christoph and Larsen, Finn and Wilczek, Frank",
    title = "{Geometric and renormalized entropy in conformal field theory}",
    eprint = "hep-th/9403108",
    archivePrefix = "arXiv",
    reportNumber = "PUPT-1454, IASSNS-HEP-93-88",
    doi = "10.1016/0550-3213(94)90402-2",
    journal = "Nucl. Phys. B",
    volume = "424",
    pages = "443--467",
    year = "1994"
}

@article{Hubeny:2007xt,
    author = "Hubeny, Veronika E. and Rangamani, Mukund and Takayanagi, Tadashi",
    title = "{A Covariant holographic entanglement entropy proposal}",
    eprint = "0705.0016",
    archivePrefix = "arXiv",
    primaryClass = "hep-th",
    reportNumber = "DCPT-07-13, KUNS-2069",
    doi = "10.1088/1126-6708/2007/07/062",
    journal = "JHEP",
    volume = "07",
    pages = "062",
    year = "2007"
}

@article{Humeniuk:2012xg,
    author = "Humeniuk, Stephan and Roscilde, Tommaso",
    title = "{Quantum Monte Carlo calculation of entanglement Renyi entropies for generic quantum systems}",
    eprint = "1203.5752",
    archivePrefix = "arXiv",
    primaryClass = "cond-mat.str-el",
    doi = "10.1103/PhysRevB.86.235116",
    journal = "Phys. Rev. B",
    volume = "86",
    pages = "235116",
    year = "2012"
}

@article{Islam:2015mom,
    author = "Islam, Rajibul and Ma, Ruichao and Preiss, Philipp M. and Tai, M. Eric and Lukin, Alexander and Rispoli, Matthew and Greiner, Markus",
    title = "{Measuring entanglement entropy through the interference of quantum many-body twins}",
    eprint = "1509.01160",
    archivePrefix = "arXiv",
    primaryClass = "cond-mat.quant-gas",
    doi = "10.1038/nature15750",
    month = "9",
    year = "2015"
}

@article{Jasser2026Antiflatness,
    author = "Jasser, Barbara and Iannotti, Daniele and Hamma, Alioscia",
    title = "{A journey through Flatland: What does the antiflatness of a spectrum teach us?}",
    eprint = "2605.21664",
    archivePrefix = "arXiv",
    primaryClass = "quant-ph",
    month = "5",
    year = "2026"
}

@article{Koenig:2009avh,
    author = "Koenig, Robert and Renner, Renato and Schaffner, Christian",
    title = "{The Operational Meaning of Min- and Max-Entropy}",
    eprint = "0807.1338",
    archivePrefix = "arXiv",
    primaryClass = "quant-ph",
    doi = "10.1109/TIT.2009.2025545",
    journal = "IEEE Trans. Info. Theor.",
    volume = "55",
    number = "9",
    pages = "4337--4347",
    year = "2009"
}

@article{Landsman:2018jpm,
    author = "Landsman, Kevin A. and Figgatt, Caroline and Schuster, Thomas and Linke, Norbert M. and Yoshida, Beni and Yao, Norman Y. and Monroe, Christopher",
    title = "{Verified Quantum Information Scrambling}",
    eprint = "1806.02807",
    archivePrefix = "arXiv",
    primaryClass = "quant-ph",
    doi = "10.1038/s41586-019-0952-6",
    journal = "Nature",
    volume = "567",
    number = "7746",
    pages = "61--65",
    year = "2019"
}

@article{Lewkowycz:2013nqa,
    author = "Lewkowycz, Aitor and Maldacena, Juan",
    title = "{Generalized gravitational entropy}",
    eprint = "1304.4926",
    archivePrefix = "arXiv",
    primaryClass = "hep-th",
    doi = "10.1007/JHEP08(2013)090",
    journal = "JHEP",
    volume = "08",
    pages = "090",
    year = "2013"
}

@article{Linden:2012kdb,
    author = "Linden, Noah and Mosonyi, Mil{\'a}n and Winter, Andreas",
    title = "{The structure of Renyi entropic inequalities}",
    eprint = "1212.0248",
    archivePrefix = "arXiv",
    primaryClass = "quant-ph",
    doi = "10.1098/rspa.2012.0737",
    journal = "Proc. Roy. Soc. Lond. A",
    volume = "469",
    number = "2158",
    pages = "20120737",
    year = "2013"
}

@article{Linke:2017xlv,
    author = "Linke, Norbert M. and Johri, Sonika and Figgatt, Caroline and Landsman, Kevin A. and Matsuura, Anne Y. and Monroe, Christopher",
    title = "{Measuring the R{\'e}nyi entropy of a two-site Fermi-Hubbard model on a trapped ion quantum computer}",
    eprint = "1712.08581",
    archivePrefix = "arXiv",
    primaryClass = "quant-ph",
    doi = "10.1103/PhysRevA.98.052334",
    journal = "Phys. Rev. A",
    volume = "98",
    number = "5",
    pages = "052334",
    year = "2018"
}

@article{Maldacena:1997re,
    author = "Maldacena, Juan Martin",
    title = "{The Large $N$ limit of superconformal field theories and supergravity}",
    eprint = "hep-th/9711200",
    archivePrefix = "arXiv",
    reportNumber = "HUTP-97-A097, HUTP-98-A097",
    doi = "10.4310/ATMP.1998.v2.n2.a1",
    journal = "Adv. Theor. Math. Phys.",
    volume = "2",
    pages = "231--252",
    year = "1998"
}

@article{Mirsky1960SymmetricGauge,
    author = "Mirsky, Leon",
    title = "{Symmetric gauge functions and unitarily invariant norms}",
    journal = "Quart. J. Math. Oxford Ser. (2)",
    volume = "11",
    pages = "50--59",
    doi = "10.1093/qmath/11.1.50",
    year = "1960"
}

@article{Mukhametzhanov:2019pzy,
    author = "Mukhametzhanov, Baur and Zhiboedov, Alexander",
    title = "{Modular invariance, tauberian theorems and microcanonical entropy}",
    eprint = "1904.06359",
    archivePrefix = "arXiv",
    primaryClass = "hep-th",
    reportNumber = "CERN-TH-2019-043",
    doi = "10.1007/JHEP10(2019)261",
    journal = "JHEP",
    volume = "10",
    pages = "261",
    year = "2019"
}

@inproceedings{Renner2004SmoothRE,
    author = "Renner, Renato and Wolf, Stefan",
    title = "{Smooth R\'enyi entropy and applications}",
    booktitle = "{Proceedings of the 2004 IEEE International Symposium on Information Theory (ISIT)}",
    pages = "232",
    doi = "10.1109/ISIT.2004.1365269",
    publisher = "IEEE",
    year = "2004"
}

@phdthesis{Renner:2005qbj,
    author = "Renner, Renato",
    title = "{Security of Quantum Key Distribution}",
    eprint = "quant-ph/0512258",
    archivePrefix = "arXiv",
    reportNumber = "Diss. ETH No. 16242, Diss. ETH No. 16242",
    school = "Zurich, ETH",
    year = "2005"
}

@article{Ryu:2006bv,
    author = "Ryu, Shinsei and Takayanagi, Tadashi",
    title = "{Holographic derivation of entanglement entropy from AdS/CFT}",
    eprint = "hep-th/0603001",
    archivePrefix = "arXiv",
    reportNumber = "NSF-KITP-06-11, NSF-KITP-06-11",
    doi = "10.1103/PhysRevLett.96.181602",
    journal = "Phys. Rev. Lett.",
    volume = "96",
    pages = "181602",
    year = "2006"
}

@article{Ryu:2006ef,
    author = "Ryu, Shinsei and Takayanagi, Tadashi",
    title = "{Aspects of Holographic Entanglement Entropy}",
    eprint = "hep-th/0605073",
    archivePrefix = "arXiv",
    reportNumber = "NSF-KITP-06-31, KUNS-2021, NSF-KITP-06-31, KUNS-2021",
    doi = "10.1088/1126-6708/2006/08/045",
    journal = "JHEP",
    volume = "08",
    pages = "045",
    year = "2006"
}

@article{Tomamichel:2009ohu,
    author = "Tomamichel, Marco and Colbeck, Roger and Renner, Renato",
    title = "{A Fully Quantum Asymptotic Equipartition Property}",
    eprint = "0811.1221",
    archivePrefix = "arXiv",
    primaryClass = "quant-ph",
    doi = "10.1109/tit.2009.2032797",
    journal = "IEEE Trans. Info. Theor.",
    volume = "55",
    number = "12",
    pages = "5840--5847",
    year = "2009"
}

@book{Tomamichel:2015gtd,
    author = "Tomamichel, Marco",
    title = "{Quantum Information Processing with Finite Resources. Mathematical Foundations}",
    eprint = "1504.00233",
    archivePrefix = "arXiv",
    primaryClass = "quant-ph",
    doi = "10.1007/978-3-319-21891-5",
    isbn = "978-3-319-21890-8, 978-3-319-21891-5",
    publisher = "Springer",
    series = "SpringerBriefs in Mathematical Physics",
    volume = "5",
    year = "2016"
}

@article{vanEnk:2011xlo,
    author = "van Enk, S. J. and Beenakker, C. W. J.",
    title = "{Measuring Tr{\ensuremath{\rho^n}} on Single Copies of {\ensuremath{\rho}} Using Random Measurements}",
    eprint = "1112.1027",
    archivePrefix = "arXiv",
    primaryClass = "quant-ph",
    doi = "10.1103/PhysRevLett.108.110503",
    journal = "Phys. Rev. Lett.",
    volume = "108",
    number = "11",
    pages = "110503",
    year = "2012"
}

@article{Wall:2012uf,
    author = "Wall, Aron C.",
    title = "{Maximin Surfaces, and the Strong Subadditivity of the Covariant Holographic Entanglement Entropy}",
    eprint = "1211.3494",
    archivePrefix = "arXiv",
    primaryClass = "hep-th",
    doi = "10.1088/0264-9381/31/22/225007",
    journal = "Class. Quant. Grav.",
    volume = "31",
    number = "22",
    pages = "225007",
    year = "2014"
}

@article{Wang:2021ptw,
    author = "Wang, Jinzhao",
    title = "{The refined quantum extremal surface prescription from the asymptotic equipartition property}",
    eprint = "2105.05892",
    archivePrefix = "arXiv",
    primaryClass = "hep-th",
    doi = "10.22331/q-2022-02-16-655",
    journal = "Quantum",
    volume = "6",
    pages = "655",
    year = "2022"
}

@article{Wilming:2018rvz,
    author = "Wilming, Henrik and Eisert, Jens",
    title = "{Single-shot holographic compression from the area law}",
    eprint = "1809.10156",
    archivePrefix = "arXiv",
    primaryClass = "quant-ph",
    doi = "10.1103/PhysRevLett.122.190501",
    journal = "Phys. Rev. Lett.",
    volume = "122",
    number = "19",
    pages = "190501",
    year = "2019"
}

@article{Winter:2015qhs,
    author = "Winter, Andreas",
    title = "{Tight Uniform Continuity Bounds for Quantum Entropies: Conditional Entropy, Relative Entropy Distance and Energy Constraints}",
    eprint = "1507.07775",
    archivePrefix = "arXiv",
    primaryClass = "quant-ph",
    doi = "10.1007/s00220-016-2609-8",
    journal = "Commun. Math. Phys.",
    volume = "347",
    number = "1",
    pages = "291--313",
    year = "2016"
}

@article{CalabreseLefevre2008Spectrum,
    author = "Calabrese, Pasquale and Lefevre, Alexandre",
    title = "{Entanglement spectrum in one-dimensional systems}",
    eprint = "0806.3059",
    archivePrefix = "arXiv",
    primaryClass = "cond-mat.stat-mech",
    doi = "10.1103/PhysRevA.78.032329",
    journal = "Phys. Rev. A",
    volume = "78",
    pages = "032329",
    year = "2008"
}

@article{Helstrom1969Detection,
    author = "Helstrom, Carl W.",
    title = "{Quantum detection and estimation theory}",
    journal = "J. Stat. Phys.",
    volume = "1",
    pages = "231--252",
    doi = "10.1007/BF01007479",
    year = "1969"
}

@article{AudenaertQuantumChernoff2007,
    author = "Audenaert, K. M. R. and Calsamiglia, J. and Munoz-Tapia, R. and Bagan, E. and Masanes, Ll. and Acin, A. and Verstraete, F.",
    title = "{Discriminating States: The Quantum Chernoff Bound}",
    eprint = "quant-ph/0610027",
    archivePrefix = "arXiv",
    primaryClass = "quant-ph",
    doi = "10.1103/PhysRevLett.98.160501",
    journal = "Phys. Rev. Lett.",
    volume = "98",
    pages = "160501",
    year = "2007"
}

@book{KhatriWilde2024Principles,
    author = "Khatri, Sumeet and Wilde, Mark M.",
    title = "{Principles of Quantum Communication Theory: A Modern Approach}",
    eprint = "2011.04672",
    archivePrefix = "arXiv",
    primaryClass = "quant-ph",
    year = "2024"
}

@book{Bhatia1997Matrix,
    author = "Bhatia, Rajendra",
    title = "{Matrix Analysis}",
    series = "Graduate Texts in Mathematics",
    volume = "169",
    publisher = "Springer",
    address = "New York",
    year = "1997",
    doi = "10.1007/978-1-4612-0653-8"
}

@article{NuradhaWilde2023Fidelity,
    author = "Nuradha, Theshani and Wilde, Mark M.",
    title = "{Fidelity-Based Smooth Min-Relative Entropy: Properties and Applications}",
    eprint = "2305.05859",
    archivePrefix = "arXiv",
    primaryClass = "quant-ph",
    journal = "IEEE Trans. Inf. Theory",
    year = "2024"
}

@article{SakaiTan2020Smooth,
    author = "Sakai, Yuta and Tan, Vincent Y. F.",
    title = "{Asymptotic Expansions of Smooth R\'enyi Entropies and Their Applications}",
    eprint = "2003.05545",
    archivePrefix = "arXiv",
    primaryClass = "cs.IT",
    journal = "IEEE Trans. Inf. Theory",
    year = "2021"
}

@article{Cohen2021DimensionFree,
    author = "Cohen, Doron and Kontorovich, Aryeh and Koolyk, Aaron and Wolfer, Geoffrey",
    title = "{Dimension-Free Empirical Entropy Estimation}",
    eprint = "2105.07408",
    archivePrefix = "arXiv",
    primaryClass = "cs.IT",
    journal = "Adv. Neural Inf. Process. Syst.",
    volume = "34",
    year = "2021"
}

@inproceedings{Cantelli1928Confini,
    author = "Cantelli, Francesco Paolo",
    title = "{Sui confini della probabilit\`a}",
    booktitle = "{Atti del Congresso Internazionale dei Matematici}",
    volume = "6",
    pages = "47--59",
    address = "Bologna",
    year = "1928"
}

@book{Boucheron2013Concentration,
    author = "Boucheron, St\'ephane and Lugosi, G\'abor and Massart, Pascal",
    title = "{Concentration Inequalities: A Nonasymptotic Theory of Independence}",
    publisher = "Oxford University Press",
    year = "2013",
    doi = "10.1093/acprof:oso/9780199535255.001.0001"
}

\end{document}